\RequirePackage{fix-cm}
\documentclass[smallextended]{svjour3}       
\smartqed  
\usepackage{graphicx}
\usepackage{xcolor}
\usepackage{algorithm}
\usepackage{algorithmic}
\usepackage{amsmath,amssymb,amsfonts}
\usepackage{multirow,longtable}
\newcommand{\todo}[1]{\textcolor{blue}{\bf #1}}
\begin{document}

\title{Landscape properties of the very large-scale and the variable neighborhood search metaheuristics for the multidimensional assignment problem
}

\titlerunning{Landscapes of the very large-scale and the variable neighborhoods 
for the MAP}

\author{Alla Kammerdiner         \and
        Alexander Semenov  \and
        Eduardo L. Pasiliao
}


\institute{Alla Kammerdiner \at
              University of Florida Research and Engineering Educational Facility (REEF), Shalimar, FL \\
              \email{alla.ua@gmail.com}           
           \and
           Alexander Semenov \at
              Herbert Wertheim College of Engineering, University of Florida, FL, USA  \\
              \email{asemenov@ufl.edu} 
                         \and
           Eduardo L. Pasiliao \at
              Munitions Directorate, Eglin Air Force Base, Eglin, FL   \\
              \email{elpasiliao@gmail.com} 
}

\date{Received: date / Accepted: date}

\maketitle
 
\begin{abstract}We study the recent metaheuristic search algorithm for the multidimensional assignment problem (MAP) using fitness landscape theory. The analyzed algorithm performs a very large-scale neighborhood search on a set of feasible solutions to the problem. 
We derive properties of the landscape graphs that represent these very large-scale search algorithms acting on the solutions of the MAP. In particular, we show that the search graph is generalization of a hypercube. 
We extend and generalize the original very large-scale neighborhood search to develop the variable neighborhood search. The new search is capable or searching even larger large-scale neighborhoods. 
We perform numerical analyses of the search graph structures for various problem instances of the MAP and different neighborhood structures of the MAP algorithm based on a very large-scale search. We also investigate the correlation between fitness (i.e., objective values) and distance (i.e., path lengths) of the local minima (i.e., sinks of the landscape). Our results can be used to design improved search-based metaheuristics for the MAP.

\keywords{combinatorial optimization \and
metaheuristics \and the multidimensional assignment problem (MAP) \and very large-scale neighborhood search (VNS) \and variable neighborhood search (VNS) 
\and
fitness landscapes}
\end{abstract}

\section{Introduction}
\label{sec:intro}
In this paper, we examine properties of the search graphs produced by very large-scale neighborhood (VLSN) search for the multidimensional assignment problem (MAP)~\cite{pierskalla1968letter}. The MAP is a combinatorial optimization problem that arises in a number of applications, including multi-sensor data association in multiple target tracking, record linkage in multiple databases, and high energy physics~\cite{kammerdiner2009multidimensional}. The MAP is known to be NP-hard~\cite{karp1972reducibility}. Furthermore, this problem is characterized by an exponential growth in input size with the dimensionality of the problem. 

Because of computational difficulty of the MAP, development of solution algorithms for the MAP has received significant amount of attention in the literature. Recently introduced very large-scale neighborhood algorithm~\cite{kammerdiner2017very} has shown some promise for solving the MAP, as it outperforms other heuristics, such as Greedy heuristic, and metaheuristics, such as genetic algorithms~\cite{kammerdiner2021multidimensional}. Efficient exploration of the set of feasible solutions of the MAP represents a key challenge in further improving the performance of VLSN-type metaheuristic algorithms for this problem.
Yet, little is known about the types of search graphs produced by the VLSN search.

 We study the very large-scale neighborhood search (VLSN) for the MAP \cite{kammerdiner2021multidimensional,kammerdiner2017very}.
%
We apply \emph{fitness landscape theory}~\cite{reidys2002combinatorial}, a mathematical framework for analysis of search algorithms on the combinatorial optimization problems (COPs). 
We establish that the VLSN landscape graph is a generalization of a 
\emph{hypercube graph}. As discussed in a review paper~\cite{schiavinotto2007review}, search landscape analysis is often used for studying how structure of the space being search influences the performance of stochastic local search algorithms. A distance between two feasible 
solutions of a given COP instance is a fundamental concept for understanding the search space structure. It is usually defined as a number of elementary operations that are needed to transform one solution into another. 

More recently, graph-based approaches~\cite{daolio2011communities,kammerdiner2013application,KAMMERDINER201478} to search landscape analysis have been introduced in addition to distance-based analysis. In particular, work~\cite{daolio2011communities} introduces a method for studying combinatorial fitness landscapes and applied it to analyze the structure of the quadratic assignment problem, another hard COP. They model a search landscape as a directed weighted graph with vertices given by the local optima and edges representing the transitions between the respective basins of attraction of two local optima.

%
In the contrast, studies
\cite{kammerdiner2013application,KAMMERDINER201478} propose a different graph-based methodology for investigating combinatorial fitness landscapes, where a search landscape is also modeled as a weighted directed graph. However, the vertices of this landscape graph include all the feasible solutions, 
not just the local optima. And the directed edges represent all possible moves from the current solution to the next one 
by a search algorithm, rather than only 
the inter-basin transitions. In our paper, we follow the graph-based landscape methodology in~\cite{KAMMERDINER201478}.

Papers~\cite{kammerdiner2013application,KAMMERDINER201478} show that the landscape digraph, which incorporates all feasible solutions, forms directed trees. Depending on whether the combinatorial optimization problem deals with minimization or maximization, a collection of directed trees in the landscape digraph is either an in-forest or an out-forest.  \cite{KAMMERDINER201478} extend and generalize the prior work  in~\cite{kammerdiner2013application}  to include out-forests and column Laplacians.
Furthermore, \cite{KAMMERDINER201478} establish a theoretically important relation between local optima and the weak components of forests. Based on this relation, they prove that the eigenvalue multiplicities of forest Laplacians give the numbers of local optima for the COPs.

The literature on combinatorial landscapes has grown significantly in the last two decades. It comprises a variety of research directions and methodologies, ranging from theoretical work in evolutionary computation~\cite{altenberg1997fitness,hordijk2005correlation,reidys2002combinatorial,richter2008coupled}, to numerical studies of search algorithms and hyper-heuristics in optimization~\cite{jones1995fitness,merz2004advanced,ochoa2009analyzing,schiavinotto2007review,stadler1992landscape}, and to recent application of landscapes in machine learning~\cite{bosman2020visualising,choromanska2015loss,choromanska2015open}. Recent advances in combinatorial landscapes are surveyed in~\cite{malan2021survey}. One of the major themes in landscape analysis remains the correlation between the solutions' fitness and distance in part because of its relation with problem instance difficulty~\cite{jones1995fitness}. 

We contribute to this literature in several ways. In particular, we are the first to consider the landscapes of the multidimensional assignment problem (MAP) that are generated by the very large-scale neighborhood search and the variable neighborhood search.
For the very large-scale neighborhood search, we establish a theoretical result regarding the graph structure of the MAP landscapes. We also present numerical comparisons between the MAP landscapes produced by alternative versions of the very large-scale neighborhood search, and the new variable neighborhood search. Last, but not least, we contribute to the large literature on the metaheuristics by describing and investigating a new metaheuristic for the multidimensional assignment, namely the variable neighborhood search.

The remainder of our paper has the following structure. In Section~\ref{sec:prelim}, we formally define basic landscape-theoretic concepts and describe the graph-based methodology in~\cite{KAMMERDINER201478}. This section also presents background knowledge on the multidimensional assignment problem (MAP) and 
the very large-scale neighborhood (VLSN) search for the MAP. 
Next, in Section~\ref{sec:hypercube} we prove that the search graph of the 
VLSN algorithm is a generalized hypercube. We propose a variable neighborhood search (VNS) for the VLSN algorithm in Section~\ref{sec:VNS}. We also connect the proposed VNS algorithm to the VLSN search graph. Importantly, Section~\ref{sec:comput} reports the results of numerical experiments aimed at
examining the structure of the VLSN search graph. The fitness-distance correlation for the VLSN search and the VNS is discussed in Section~\ref{sec:correl}. Finally, Section~\ref{sec:end} concludes the paper.

\section{Preliminaries}\label{sec:prelim}
\subsection{Combinatorial Fitness Landscapes}
Fitness landscape analysis is a valuable methodology for investigating metaheuristics in combinatorial optimization. An application of a search algorithm to solving an instance of the COP induces a structure on the space of this COP's solutions. The purpose of fitness landscapes theory is to study how the search behavior and performance are impacted by the underlying structure of the solutions' space. 

In this section, we formally define a fitness landscape. In particular, we first define related concepts such as a neighborhood of a solution. Next, we present a 
definition of fitness landscape. 
 The COP's fitness landscape can be viewed 
    as a directed graph
Then we discuss the relation between features of a landscape graph and local optima of the COP.
Specifically, 
there 
is a correspondence between sinks (sources) of directed in-(out-)trees and the COP's local minima (maxima).
\begin{figure}
    \centering
    \includegraphics[width = 0.75\textwidth]{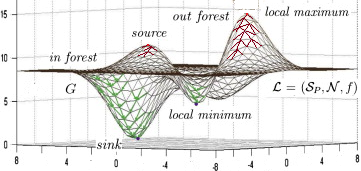}
    \caption{A landscape $\mathcal{L}$ of an instance $P$ of some combinatorial optimization problem with an objective function $f$ and a neighborhood mapping $\mathcal{N}$}
    \label{fig:landscape_JOGO}
\end{figure}

Following~\cite{KAMMERDINER201478}, we adopt the terminology of~\cite{schiavinotto2007review}. 
We consider a COP such as the MAP, with other examples of COPs whose landscapes have been studied being the travelling salesman problem (TSP) and the quadratic assignment problem (QAP). Let $P$ be an instance of the COP with a set of feasible solutions $S_P$. Then $\mathbf{2}^{S_P}$ denotes the set of all feasible solutions of instance $P$. There is a neighborhood mapping $\mathcal{N}:S_P \rightarrow \mathbf{2}^{S_P}$ that assigns to each solution $s$ its subset $\mathcal{N}(s)$ of solutions, called \emph{neighbors} $s^{\prime}\in\mathcal{N}(s)$. For convenience, we assume that  $s \notin \mathcal{N}(s)$.
 \vspace{10pt}

In combinatorial landscapes, a search algorithm is described and studied with respect to a distance between solutions and the solution quality or fitness. In the context of combinatorial landscapes, a neighborhood is defined using an operator $\Delta$ that determines all allowable moves for a single step of the search algorithm. 
 Given a single-step operator $\Delta$, a \emph{distance} $d(s,t)$ between two solutions $s,t$ is the minimum number of 
 steps $\delta \in \Delta$ needed to reach $s$ from $t$ or $t$ from $s$.
Given a problem instance $P$ and the objective values of its solutions, the 
\emph{fitness function} $f:S_P \rightarrow \mathbf{R}$ is a mapping from a feasible solution to its objective value.

Now that we have defined fitness and distance with respect to a search algorithm, we can introduce a key concept of combinatorial fitness landscapes.
A \emph{fitness landscape} of COP instance $P$ is a triple $(S_P, \mathcal{N}, f)$ that consists of
\begin{enumerate}
    \item[(i)] a set $S_P$ of all feasible solutions of $P$;
    \item[(ii)] a neighborhood $\mathcal{N}$ on $S_P$ that corresponds to operator $\Delta$ and distance $d$ between the solutions in $S_P$;
    \item[(iii)] a fitness function $f$ of $P$ that assigns its objective value $f(s)$ to every solution $s\in S_P$.
\end{enumerate}

Next, we describe a digraph representation of combinatorial landscapes, which was introduced in~\cite{KAMMERDINER201478} building on the diagraph methodology in~\cite{agaev2000matrix,agaev2005spectra}. 
Let $G=(V,E)$ denote a weighted digraph without loops with a vertex set $V  = V(G) = \{1,\ldots,n\}$ that consists of $n$ nodes, and n edge set $E = E(G) = \{ e_k = (i_k,j_k) : i_k,j_k\in V(G), k= 1,\ldots,m \}$ comprised of $m$ directed edges. Without loss of generality, we  assume that edges $e_k \in E$ 
    have strictly positive weights $w(e_k)>0$  for all $k=1,\ldots,m$.
We say that a weighted digraph $G=(V,E)$ represents a combinatorial landscape if 
\begin{enumerate}
        \item[(i)]  $V := S_P$, i.e., the vertex set consists of all feasible solutions;
        \item[(ii)] $E := \{ (i,j): i,j\in S_P, \quad j \in \mathcal{N}(i), \quad f(i) > f(j)  \}$, or the directed edges link neighboring solutions from a lower quality solution to a higher quality solution; 
        \item[(iii)] the weight of any directed edge $e=(i,j)\in E$ is given by $w(e) := f(i) - f(j) >0$. 
    \end{enumerate}
Moreover, we say that an objective function and a fitness landscape are \emph{disjunctive} if $f(i) \neq f(j)$ for any $i \neq j$. In~\cite{KAMMERDINER201478}, the correspondence between local minima $s^{\ast}$ of a problem instance $P$ and the sinks of landscape graph $G$ is established. Furthermore, there exists an underlying directed tree structure for the landscape. This is illustrated in Figure~\ref{fig:landscape_JOGO}.

\subsection{The Multidimensional Assignment Problem}\label{sec:MAP}
The multidimensional assignment problem (MAP) is an extension of the well-known linear assignment problem (LAP). In general, the MAP is NP-hard, unlike the LAP, which can be solved in polynomial time (e.g., by Hungarian algorithm). The linear assignment seeks to find a matching between two subsets (e.g., one-to-one assignment of jobs to machines) with the least costs.
An extension of the LAP to three of more dimensions, the MAP seeks to find a matching among multiple subsets (e.g., jobs, machines, time periods). 

The MAP has many alternative formulations, including 0-1 integer programming problem~\cite{pierskalla1967tri}, and a problem of partitioning a vertex set of a complete $D$-partite graph into $N$ pairwise disjoint cliques~\cite{burkard1999linear}. Because the VLSN search algorithm takes advantage of the formulation of the MAP using combinatorial optimization, that is the one we present here. The decision 
problem is to find a $D$-way one-to-one assignment of $N$ observations in $D$ subsets.

The multidimensional assignment problem can be written using its combinatorial optimization formulation as follows:
\begin{equation}\label{eq:MAP}
 \min_{\pi_2,\ldots,\pi_D \in \Pi_N} \sum_{i=1}^{N} c_{i \pi_2(i) \ldots \pi_D(i)},
\end{equation}
where $\Pi_N = \{ (\pi(i))_{i=1}^{N} \}$ is the set of all possible permutations of $\{1,2,\ldots,N\}$, a vector of permutations $(\pi_2,\ldots,\pi_D )$ is a feasible solution of the MAP, and $c_{i i_2 \ldots i_D}$ with $i,i_k = 1,\ldots,N$, $k=2,\ldots,D$ is a given array of cost coefficients.
 The LAP matches $D=2$ dimensions, and so the LAP can be written in an analogous fashion:
\begin{equation}\label{eq:LAP}
 \min_{\pi_2 \in \Pi_N} \sum_{i=1}^{N} c_{i \pi_2(i) },
\end{equation}
where permutation $\pi_2$ is a feasible solution of the LAP and $C = (c_{ij})$, $i,j=1,\ldots,N$.

The array of cost coefficients serves as an input. Observe that the input size is $N^D$. It grows exponentially in $D$ and at least cubically $N^3$  in $N$, since $D\geq3$. Rapid growth of input size with the MAP parameters of cardinality $N$ and dimensionality $D$ is illustrated in Figure~\ref{fig:MAPinput}. The rows in the figure denote different values of parameter $D$, while the columns represent possible values of $N$. The input size reaches terabyte scale for moderate values of $D$ and $N$. 
\begin{figure}
  \centering
  \includegraphics[angle=0,width=0.9\textwidth]{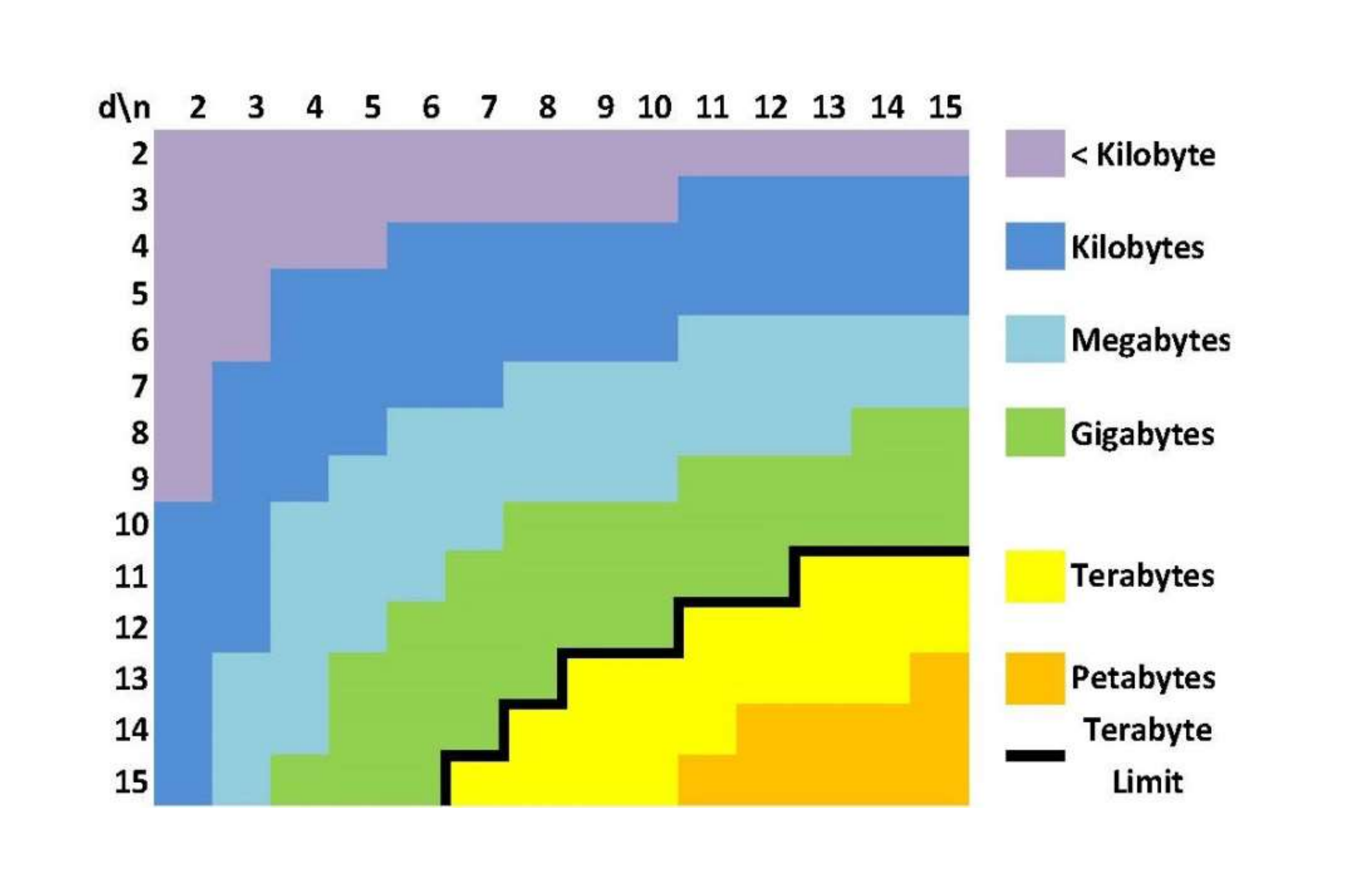}\\
  \caption{The input data for the MAP quickly reaches terabyte scale}
  \label{fig:MAPinput}
\end{figure}

\subsection{The Very Large-Scale Neighborhood Search for the MAP}
\label{sec:vlsn}

In this section, we present the VLSN search algorithm for the MAP. This algorithm was recently proposed in~\cite{kammerdiner2017very}. The VLSN search evaluates exponentially large solution neighborhoods by solving suitably chosen LAP instances. Algorithm~\ref{alg:VLSN.OP} shows the pseudocode for the VLSN search. First, we introduce some notations used in the pseudo code to explain how the procedure works. Then we define the projection matrix and describe how it is computed from the array of the MAP cost coefficients. Finally, we explain the procedure in Algorithm~\ref{alg:VLSN.OP}.

Consider an instance of the MAP~\eqref{eq:MAP} with cardinality $N$, dimensionality $D$, and $N^D$ array of cost coefficients $C = (c_{i_1 \ldots i_D})$, $i_k = 1,\ldots,N, k=1,\ldots, D$. Suppose the search algorithm employs  a multi-start strategy  $\Sigma=\{s_1,s_2,\ldots,s_{\mu}\}$ which consists of $\mu$ initial
      feasible solutions. These initial solutions can be constructed in some way, for example, by randomly generating $D - 1$ permutations $\pi_2,\ldots,\pi_{D}$ of size $N$.
Let $S = (\pi_2,\ldots,\pi_{D})$ denote a feasible solution of the MAP with objective value $y = y(S) := \sum_{i=1}^N c_{i \pi_2(i) \ldots \pi_{D}(i)}$.
%
%
%
Let $\textrm{LAP}\left(B\right) = (\pi,y(\pi))$ 
       be the optimal solution $\pi\in \Pi_N$ with objective value $y(\pi) = \sum_{i=1}^N b_{i \pi(i)}$ of an instance of the LAP with coefficient matrix $B = (b_{ij})_{i,j=1}^{N}$. 

Next, we define the projection matrix. Given the current feasible solution $S$ and a dimension $d$ of the MAP, the projection matrix $C(d,S)$ is constructed by suitably choosing $N\times N$ coefficients from the MAP cost array $C$. Note that we use two different (but analogous) ways to construct this matrix dependent on whether or not a given dimension is $d=1$. 
A modified way of constructing the projection is needed for the first dimension 
because all feasible solutions are written assuming an identical permutation in dimension $d=1$ (otherwise solutions would not be unique). In fact, a feasible solution of the MAP has a form of $(\pi_2,\ldots,\pi_D) = \left( \begin{array}{ccc}
    \pi_2(1) & \ldots & \pi_D(1) \\
    \vdots & \ddots & \vdots \\
    \pi_2(N) & \ldots & \pi_D(N)
\end{array} \right)$, a $N\times(D-1)$ matrix that consists of $D-1$ permutations $\pi_2,\ldots,\pi_D$ in columns $j=1,\ldots,D-1$.
\begin{definition}\label{def:proj1}
    Let $d \in \{1,\ldots,D\}$ be a given dimension, and $S=(\pi_2,\ldots,
    \pi_D)$ be the current solution of the MAP.  
        The projection of the solution $S$ onto the MAP cost array 
        $C$ 
        along %
        dimension $d$
        is an $N \times N$ matrix $C(d,S)$
        with elements $C_{ij}(d,S)$, $i,j=1,\ldots,N$, defined as:
        \begin{enumerate}
          \item[(i) ] 
          $C_{ij}(d,S)=c_{i \pi_2(i) \ldots j \ldots \pi_{D}(i)}$
          when $d\neq 1$,
          \item[(ii)] 
          $C_{ij}(1,S)=c_{i \pi_2(j) \ldots \pi_{D}(j)}$
          when $d=1$.
        \end{enumerate}
    \end{definition}

The projection matrix $C(d,S)$ is used to quickly search through $N!$ neighborhood solutions. These solutions differ from the current solution $S$ only by a permutation in one dimension (i.e., $d$). 
Solving the LAP~\eqref{eq:LAP} with cost coefficients $C(d,S)$ produces the optimal permutation $\pi^{\ast}_d$ that replaces the current $\pi_d$ in $S$ to create a new and improved current solution $S^{\prime}$. When the neighborhood in the first dimension is searched, the first dimension must remain fixed to the identity permutation.
 So, rather than searching for the optimal permutation in the first dimension, the projection matrix $C(1,S)$ is defined to find the permutation of the remaining dimensions together with respect to the first dimension.  
 
 Overall, Algorithm~\ref{alg:VLSN.OP} works as follows. It uses $\mu$ starting solutions. For each starting solution, the algorithm proceeds searching through very large-scale neighborhoods until a local optimum is reached (i.e., we have no further improving move). The search through the neighborhoods is repeated for each dimension $k$ by first computing the costs matrix $C(k,S)$ for the LAP as the ``projection'' defined above, and then solving the LAP. The LAP solution returns the permutation $\pi^{\ast}_k$ and objective value $y^{\ast}_k$. The basic, steepest descent version of Algorithm~\ref{alg:VLSN.OP} chooses the dimension $k$ and permutation $\pi^{\ast}_k$ with the smallest value $y^{\ast}_k$. If the update improves on $S$, the solution $S$ is updated by replacing the permutation in the $k$-th dimension with $\pi^{\ast}_k$. Otherwise the local minimum is reached, and the algorithm restarts in the next initial solution unless all initial solutions have been exhausted. 
 

\begin{algorithm}[h!]
\caption{Very Large-Scale Search via Optimal Permutation}
\label{alg:VLSN.OP}
\begin{algorithmic}[1]
\STATE Input $\mu$, $D$, $N$, 
$C = \{ c_{i_1 \ldots i_D} | i_k=1,\ldots,N, 
k=1,\ldots,D \}$
\STATE Input or Generate $\mu$ starting solutions 
$\Sigma= \{S_1,\ldots,S_{\mu} \}$ 
\FORALL{$S=(s_2,\ldots,s_{D})\in \Sigma$}
\STATE Compute objective value $y\leftarrow y(S)$ via Equation~(\ref{eq:MAP})
\REPEAT
    \FOR{ dimension $k=1$ \TO $D$} 
    \STATE Get projection $C(k,S)$ 
    of multidimensional matrix $C$ 
    \STATE Solve $\textrm{LAP}\left(C(k,S)\right) = (\pi_k^{\ast},y_k^{\ast})$ 
    \ENDFOR
    \STATE Find $k^{\prime}$ such that $y_{k^{\prime}} = \min\{y_k^{\ast}: k=1,\ldots,D\}$ 
    \IF{$y_{k^{\prime}} < y$} 
        \STATE Update 
        $S \leftarrow (s_2,\ldots,s_{k^{\prime}-1},\pi_{k^{\prime}}, s_{k^{\prime}+1}, \ldots, s_{D})$ 
        \STATE $y \leftarrow y_{k^{\prime}}$
    \ENDIF
\UNTIL{no improving move is left}
\ENDFOR
\STATE Output a feasible solution $S$ of the MAP~\eqref{eq:MAP} and cost $y$
\end{algorithmic}
\end{algorithm}


\section{VLSN Graph as a Generalized Hypercube}
\label{sec:hypercube}
In this section, we study the landscape of the VLSN search for the MAP with $N=2$. In particular, we show that an undirected graph of this landscape can be represented by a regular hypercube. 

First, we consider an example of the VLSN landscape graph for the MAP with $N=2$ and $D=3$.
An instance of the MAP with these parameters have the total number of solutions  $N_{sol} = (N!)^{D-1} = 2^2 = 4$.
Let these four feasible solutions be denoted by $\textbf{I} = \left(
            \begin{array}{cc}
                 1 & 1 \\
                 2 & 2
            \end{array}
        \right)$, 
        $\mathbf{\Pi} = \left(
            \begin{array}{cc}
                 2 & 2 \\
                 1 & 1
            \end{array}
        \right)$, 
         $\textbf{A} = \left(
            \begin{array}{cc}
                2 & 1 \\
                1 & 2
            \end{array}
        \right)$, 
        $\textbf{B} = \left(
            \begin{array}{cc}
                 1 & 2 \\
                 2 & 1
            \end{array}
        \right)$. 
        These solutions form a node set of the landscape graph.
The objective values of these solutions are $y_I = c_{111} + c_{222}$, $y_{\Pi} = c_{122} + c_{211}$, $y_A = c_{121} + c_{212}$, and $y_B  = c_{112} + c_{221}$. The direction of the edges between solution nodes  \textbf{I}, $\mathbf{\Pi}$, \textbf{A}, \textbf{B} depends on whether the respective differences $y_i - y_j$ are positive or negative. While the presence of an edge is determined based on how the neighborhood is defined. We consider two alternative definitions of a neighborhood, i.e., (i) the neighborhood that includes the search through all dimensions, and (ii) the neighborhood that excludes the search in the first dimension.   
Suppose the VLSN algorithm searches in only $D-1$ dimensions $2,\ldots,D$ and does not search the first dimension. Then node \textbf{I} ($\mathbf{\Pi}$) is connected with edges to \textbf{A} and \textbf{B}. These connections are shown in the left subplot of Fig.~\ref{fig:vlasn_n2d3} for $d = 2,3$. This subplot is a 2-dimensional hypercube. %
When the first dimension $d=1$ is also included, making the VLSN search in all dimensions $1,2,\ldots, D$, then the diagonals $(\mathbf{I}, \mathbf{\Pi})$ and $(\mathbf{A},\mathbf{B})$ are added. These additional edges are shown in the right subplot of Fig.~\ref{fig:vlasn_n2d3} for $d = 1,2,3$. This right subplot is a 2-dimensional hypercube with two diagonals. 

\begin{figure}
    \centering
    \includegraphics[width = 0.25\textwidth]{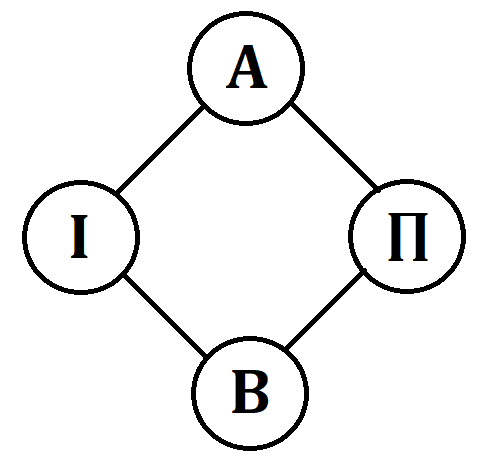}
    $\qquad \qquad \qquad$ 
    \includegraphics[width = 0.25\textwidth]{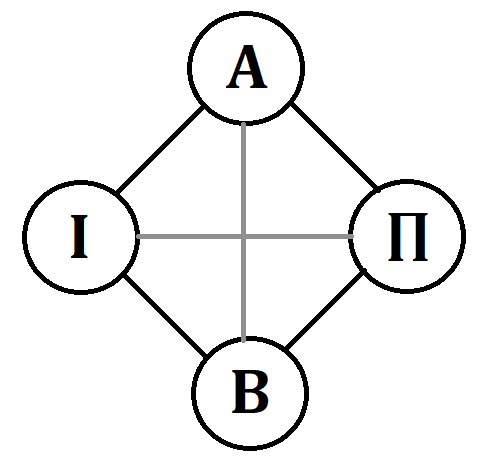}
    \caption{Undirected graphs of the VLSN search in $D-1 = 2$ dimensions $2,3$ (left) and $D=3$ dimensions $1,2,3$ (right) for the MAPs with $N=2, D=3$}
    \label{fig:vlasn_n2d3}
\end{figure}

Next, consider another example of the 
VLSN landscape graph
for MAP with $N=2$ and $D=4$. 
Then $N_{sol} = (N!)^{D-1} = 8$ is the number of feasible solutions for an instant with these parameters.
We follow as similar notation as in the first example to represent these 8 solutions.
Specifically, we denote
        $\textbf{ I} = \left(
            \begin{array}{ccc}
                 1 & 1& 1 \\
                 2 & 2& 2
            \end{array}
        \right)$ \qquad \,\, $\mathbf{\Pi} = \left(
            \begin{array}{ccc}
                2 & 2& 2 \\
                1 & 1& 1
            \end{array}
        \right)$
        \qquad $\textbf{ A} = \left(
            \begin{array}{ccc}
                2 & 1& 1 \\
                1 & 2& 2
            \end{array}
        \right)$ \qquad $\textbf{ B} = \left(
            \begin{array}{ccc}
                1 & 2& 1 \\
                2 & 1& 2
            \end{array}
        \right)$
        \qquad  $\textbf{ C} = \left(
            \begin{array}{ccc}
                1 & 1& 2 \\
                2 & 2& 1
            \end{array}
        \right)$ \qquad $\textbf{AB} = \left(
            \begin{array}{ccc}
                2 & 2& 1 \\
                1 & 1& 2
            \end{array}
        \right)$
        \qquad $\textbf{AC} = \left(
            \begin{array}{ccc}
                2 & 1& 2 \\
                1 & 2& 1
            \end{array}
        \right)$ \qquad $\textbf{BC} = \left(
            \begin{array}{ccc}
                1 & 2& 2 \\
                2 & 1& 1
            \end{array}
        \right)$.
These 8 solutions form a vertex set of the landscape graph. The objective values of these solutions are  $y_{I} = c_{1111} + c_{2222}$, $y_{\Pi} = c_{1222} + c_{2111}$, $y_A = c_{1211} + c_{2122}$, $y_B  = c_{1121} + c_{2212}$, $y_{C} = c_{1112} + c_{2221}$, $y_{AB} = c_{1221} + c_{2112}$, $y_{AC} = c_{1212} + c_{2121}$, and $y_{BC}  = c_{1122} + c_{2211}$. 
The objective values determine the orientation of edges.    
While the VLSN search and the related neighborhoods determine the existence of edges.
Again, We consider two definitions of a neighborhood, i.e., (i) the one that excludes the VLSN search in the first dimension, and (ii) the other one based on the search through all dimensions. %
First, assume the VLSN algorithm searches in only $D-1$ dimensions $2,\ldots,D$ and does not search the first dimension. Then the edge set is $V_{2,3,4} = \{ (\mathbf{I}, j) : \quad j \in\mathbf{A},  \mathbf{B},\mathbf{C}  \} \cup \{ (\mathbf{\Pi}, j) : \quad j \in\mathbf{AB},  \mathbf{BC},\mathbf{AC}  \} \cup \{ (\mathbf{A},j) : \quad j \in\mathbf{AB},\mathbf{AC}) \} \cup \{ (\mathbf{B},j) : \quad j \in\mathbf{AB},\mathbf{BC}) \} \cup \{ (\mathbf{C},j) : \quad j \in\mathbf{AC},\mathbf{BC}) \}$. These connections are shown in the left subplot of Fig.~\ref{fig:vlasn_n2d4} for $d = 2,3,4$. This subplot is a 3-dimensional hypercube, where $D-1 =3$. %
When the first dimension $d=1$ is also included, making the VLSN search in all dimensions $1,2,\ldots, D$, then the 3-dimensional diagonals $(\mathbf{I}, \mathbf{\Pi})$, $(\mathbf{A},\mathbf{BC})$, $(\mathbf{B}, \mathbf{AC})$ and $(\mathbf{C},\mathbf{AB})$ are added. These additional edges are shown in the right subplot of Fig.~\ref{fig:vlasn_n2d4} for $d = 1,2,3, 4$. This right subplot is a 3-dimensional hypercube with four diagonals. 

\begin{figure}
        \centering
        \includegraphics[width = 0.3\textwidth]{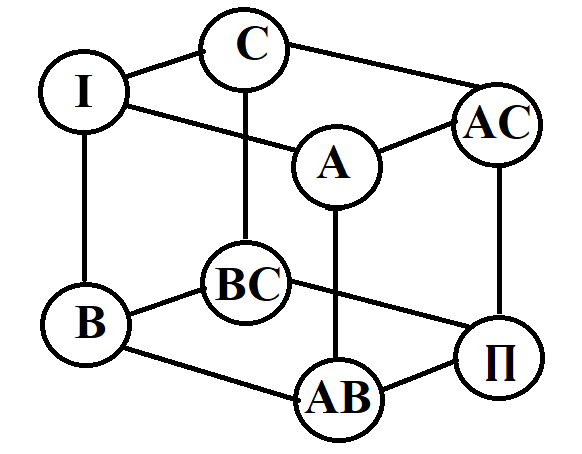}
        $\quad$
        \includegraphics[width = 0.3\textwidth]{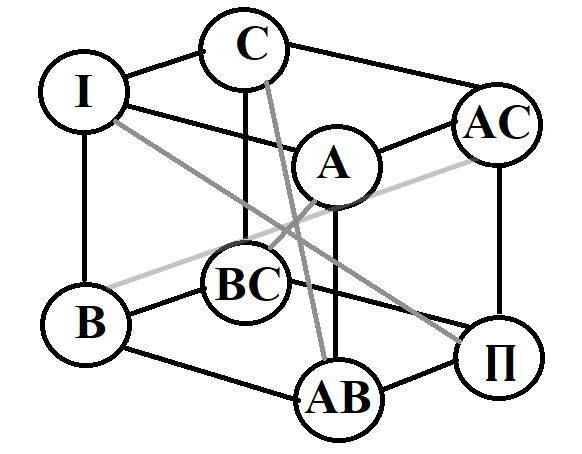}
        \caption{Undirected graphs of the VLSN search in $D-1 =3$ dimensions $2,3,4$ (left) and in $D=4$ dimensions $1,2,3,4$ (right) for the MAPs with $N=2, D=4$}
        \label{fig:vlasn_n2d4}
    \end{figure}

\begin{proposition}
For the MAP with $N=2$, the VLSN landscapes are hypercubes.
Specifically, the undirected graphs are $(D-1)$-regular hypercubes for VLSN in $D-1$ dimensions. 
The $N^{D-2}$ ``main'' diagonals are added into a hypercube for the VLSN in all $D$ dimensions. 
\end{proposition}
\begin{proof}
By induction on the dimensionality $D$. Consider the VLSN search in $D - 1$ dimensions. In fact, the claims are verified for $D=3,4$. Suppose the claims hold for some $D_0 = D_1-1$. Then the undirected graph $G_0$ for the MAP with $N=2$ and $D = D_0 = D_1-1$ contains $(N!)^{D_1-2}$ vertices that represents feasible solution. For any solution vertex $s_0 = (\pi_1,\ldots,\pi_{D_1-1})$, we can construct two solution vertices $s_1^{(1,2)} = (\pi_1,\ldots,\pi_{D_1-1},\left(\begin{array}{c}
    1  \\
    2 
\end{array}\right))$ and
 $s_1^{(2,1)} = (\pi_1,\ldots,\pi_{D_1-1},\left(\begin{array}{c}
    2  \\
    1 
\end{array}\right))$ that belong to the undirected graph $G_1$ for the MAP with $N=2$ and $D = D_1$.
By construction, the edges of $G_0$ form the edges between vertices of the form $s_1^{(1,2)}$ in $G_1$. Similarly,  the edges of $G_0$ form the edges between vertices of the form $s_1^{(2,1)}$ in $G_1$.
There is an edge between the corresponding vertices  $s_1^{(1,2)}$ and  $s_1^{(2,1)}$ in $G_1$ because they only differ in dimension $d=D_1$, and one of them has a better (smaller) objective value than the other.
Since $G_0$ is a $(D_0-1)$-regular hypercube (by assumption), we have two $(D_0-1)$-regular hypercubes on the vertices  $s_1^{(1,2)}$ and  $s_1^{(2,1)}$, connected by each other by an edge $(s_1^{(1,2)},  s_1^{(2,1)})$. Hence, these vertices and edges together form a  $(D_1-1)$-regular hypercube. Analogously, one can show that the VLSN search in $D$ dimensions adds $N_{D-2}$ main diagonals to this hypercube.
\end{proof}


\section{Variable Neighborhood Search for the VLSN algorithm}
\label{sec:VNS}
Here we present a new VLSN-based metaheuristic that implements a variable neighborhood search through very large-scale neighborhoods of increasing size.
Specifically, we propose the variable neighborhood search (VNS) for the MAP that extends VLSN search. Basically, the original VLSN search, which is described in Section~\ref{sec:vlsn}, worked by considering a \emph{single} dimension $d$ and  finding an optimal permutation $\pi^{\ast}$ to replace in the $d$-th column $s_d$ of the MAP solution $S=(s_2,\ldots,s_d,\ldots,s_D)$, while keeping the remaining $D-2$ permutation columns $\{s_2,\ldots,s_D\}\setminus \{s_d\}$ unchanged/fixed. In the new VNS, expanding neighborhoods can be used by considering \emph{multiple} dimensions $d_1,\ldots, d_K$ instead of the single dimension. For the single dimension in the VLSN search, $D$ (or $D-1$) alternative choices exist (depending on the VLSN search implementation, with or without permutations in the very first dimension, i.e., with or without diagonals in the landscape). Whereas for the two dimensions in the VNS, there exists $\left( \begin{array}{c}
    D \\
    2 
\end{array}\right) = D(D-1)/2$ (or $(D-1)(D-2)/2$) ways to choose two dimensions. In general, if the VNS considers $k$ dimensions (where $k<D$), then the VNS neighborhood consists of $\left( \begin{array}{c}
    D \\
    k 
\end{array}\right) = D(D-1)\cdot\ldots\cdot(D-k)/k!$ (or $(D-1)\cdot\ldots\cdot(D-k-1)/k!$) ways to move to a neighboring solution. The trade-off is that the VLSN search is linear in the number of considered dimensions, and hence at most $D$ of the linear assignments or matching problems solved. While the VNS is polynomial with $O(D^k)$ for any specific $k$-dimensional choice of an expanded neighborhood, and the VNS is exponential when all possible $k$-subsets of considered dimensions are included into the VNS neighborhood.

The idea of searching in a subset of dimensions rather than a single one is not entirely new for the MAP. Neither is the general VNS, which is an established class of metaheuristics. A VNS algorithm typically start with some standard neighborhood, and keeps searching until local minima are reached, when no more improvement is possible, the use of the VNS allows expanding the neighborhood to try to escape the local minimum. 
What is new is how the VLSN search is extended to form the VNS. 

To extend the VLSN search from a single dimension in the original algorithm (Alg.~\ref{alg:VLSN.OP}) to multiple dimensions in the proposed VNS, we first 
generalize
the definition of a projection. While Definition~\ref{def:proj1} introduces a concept of the projection along a single dimension, Definition~\ref{def:proj2_k_dim} below gives the projection along many dimensions, specifically some proper subset of $k$ out of $D$ dimensions (or $D-1$ for the VLSN search implementation without permuting the first dimension). Def.~\ref{def:proj2_k_dim} is a generalization of Def.~\ref{def:proj1}, because when $k=1$ (or $k=D-1$) both definitions are essentially identical.

\begin{definition}\label{def:proj2_k_dim}
    Consider $k$ such that $k<D$. Let $d_1, \ldots, d_k \in \{1,\ldots,D\}$, such that $d_i \neq d_j$ for $i \neq j$, be some given dimensions. Let an $N\times (D-1)$ matrix $S=(\pi_2,\ldots,
    \pi_D)$ with permutations as its $D-1$ columns be the current solution of the MAP instance $P$. Let an $N^D$ array $C$ denote the array of cost coefficients $c_{i i_2 \ldots i_D} := c\;[i, i_2, \ldots, i_D] $ in the objective function of the MAP~\ref{eq:MAP}.
        The projection of the solution $S$ onto the array $C$ 
        along %
        dimensions $d_1,\ldots,d_k$
        is an $N \times N$ matrix $C(d_1,\ldots,d_k,S)$
        with elements $C_{ij}(d_1,\ldots,d_k,S)$, $i,j=1,\ldots,N$, defined as:
        \begin{enumerate}
          \item[  (i)] 
          $C_{ij}(d_1,\ldots,d_k,S)=
          c\;[i, \pi_2(i), \ldots, \mathbf{\pi_{d_1}(j), \ldots, \pi_{d_k}(j)}, \ldots, \pi_{D}(i)]
          $
          when 
          $d_i\neq 1$ for all $i=1,\ldots,k$;
          \item[ (ii)] 
          $C_{ij}(1,d_1,\ldots,d_k,S)=
          c\;[i, \pi_2(j), \ldots, \mathbf{\pi_{d_1}(i), \ldots, \pi_{d_k}(i)}, \ldots, \pi_{D}(j)]
          $
          when there exists $d\in\{d_1,\ldots,d_k\}$ such that $d=1$.
        \end{enumerate}
    \end{definition}


Now the notion of projection is generalized to a $k$-subset of dimensions. Next, analogously to the VLSN search in one-dimension, we extend the VLSN to the VNS, which is essentially the VLSN search along $k$ 
dimensions:
        \begin{enumerate}
            \item Choose a subset of $k$ out of $D-1$ (or $D$) dimensions $d_1,\ldots,d_k$;
            \item Holding the assignment of $(\pi_{d_1},\ldots, \pi_{d_k})$ fixed in the current solution $S$:
            \begin{enumerate}
                \item Get the projection costs matrix $C^\prime = C(d_1,\ldots,d_k,S)$;
                \item Solve the LAP for $C^\prime$.
            \end{enumerate}
        \end{enumerate}

Algorithm~\ref{alg:VLSN.VNS} below presents the variable neighborhood search (VNS) for the VLSN search. Choose an (integer) order $K$ of the VNS such that $K<D$. This algorithms works by splitting $D$ dimensions into two subsets of dimension indices $(K,D-K)$ such that $K \leq D-K$. As in Alg.~\ref{alg:VLSN.OP}, $\mu$ denote the number of multi-starts (i.e., the number of initial starting solutions) in Alg.~\ref{alg:VLSN.VNS}. The procedure in Alg.~\ref{alg:VLSN.VNS} considers two types of dimension subsets: (i) all possible combinations of $K$ dimensions out of the $D-1$ dimensions that exclude the first dimension; and (ii) all possible combinations of $K-1$ dimensions out of the same $D-1$ dimensions excluding the first (which is the same as considering combinations of $D-K-1$ dimensions out of $D-1$). Similar to Alg.~\ref{alg:VLSN.OP}, this is done because the first dimension must remain fixed to an identity permutation to avoid double-counting the feasible solutions of the MAP (that is why the solution matrices contain only permutation columns of the last $D-1$ dimensions).


\begin{algorithm}[h!]
\caption{Order-$K$ Variable Neighborhood Search 
for the 
VLSN Search}
\label{alg:VLSN.VNS}
\begin{algorithmic}[1]
\STATE Input $K$, $\mu$, $D$, $N$, 
$C = \{ c_{i_1 \ldots i_D} | i_d=1,\ldots,N, 
d=1,\ldots,D \}$
\STATE Input or Generate $\mu$ starting solutions 
$\Sigma= \{S_1,\ldots,S_{\mu} \}$ 
\FORALL{$S=(s_2,\ldots,s_{D})\in \Sigma$}
\STATE Compute objective value $y\leftarrow y(S)$ via Equation~(\ref{eq:MAP})
\REPEAT
    \STATE \COMMENT{all $K$ dimensions are in the $(D-1)$-element index set $\{2,\ldots,D\}$}
    \FOR{ choice $h=1$ \TO $H = \left(\begin{array}{c}
         D-1  \\
         K 
    \end{array}\right)$}
    \STATE Choose $K$ dimensions $k_1,\ldots,k_K$ out of $D-1$
    \STATE Observe the remaining $D-K$ dimensions $1,d_1, \ldots,d_{D-K-1}$
    \STATE Get projection $$C(h,S) 
    := C(k_1,\ldots,k_K,S)
    =: C(k_1,\ldots,k_K;1,d_1, \ldots,d_{D-K-1};S)$$ 
    of 
    cost array $C$ using Def.~\ref{def:proj2_k_dim}
    \STATE Solve $\textrm{LAP}\left(C(h,S)\right) = (\pi_h^{\ast},y_h^{\ast})$ 
    \ENDFOR
    
    \STATE \COMMENT{one of $K$ dimensions is $d=1$, and the rest  $K-1$ are in the $(D-1)$-element index set $\{2,\ldots,D\}$}
    \FOR{ choice $h0=1$ \TO $H0 = \left(\begin{array}{c}
         D-1  \\
         K-1 
    \end{array}\right)$}
    \STATE Choose $K -1$ dimensions $k_1,\ldots,k_{K-1}$ out of $D-1$
    \STATE Observe the remaining $D-K$ dimensions $d_1, \ldots,d_{D-K}$
    \STATE Get projection 
    $$C(h0,S) := C(1,k_1,\ldots,k_{K-1},S)
    =: C(1,k_1,\ldots,k_{K-1};d_1, \ldots,d_{D-K};S)$$ 
    of 
    cost array $C$ using Def.~\ref{def:proj2_k_dim} 
    \STATE Solve $\textrm{LAP}\left(C(h0,S)\right) = (\pi_{h0}^{\ast},y_{h0}^{\ast})$ 
    \ENDFOR
    
    \STATE Find $h^{\prime}$ such that $y_{h^{\prime}} = \min\left\{\min\{y_{h}^{\ast}: h=1,\ldots,H\},\min\{y_{h0}^{\ast}: h0=1,\ldots,H0\}\right\}$ 
    \IF{$y_{h^{\prime}} < y$}
        \STATE $y \leftarrow y_{h^{\prime}}$
    \IF{$h^{\prime} = \min\{y_{h}^{\ast}: h=1,\ldots,H \}$}
        \STATE Update $S \leftarrow \left(s_2,\ldots,\pi^{\ast}_h(s_{k_1}),\ldots,\pi^{\ast}_h(s_{k_K}), \ldots, s_{D}\right)$ 
    \ENDIF
    \IF{$h^{\prime} = \min\{y_{h0}^{\ast}: h0=1,\ldots,H0\}$ }
        \STATE Update $S \leftarrow \left(s_2,\ldots,\pi^{\ast}_{h0}(s_{d_1}),\ldots,\pi^{\ast}_{h0}(s_{d_{D-K-1}}), \ldots, s_{D}\right)$ 
    \ENDIF
    \ENDIF
\UNTIL{no improving move is left}
\ENDFOR
\STATE Output a feasible solution $S$ of the MAP~\eqref{eq:MAP} and cost $y$
\end{algorithmic}
\end{algorithm}

\section{Computational Results for the VLSN Search Graph}
\label{sec:comput}
In this section, we present some computational results for the landscapes of the VLSN search and VNS metaheuristics for the MAP. Both of these algorithms can alternate between exploration (i.e., multi-start) and exploitation (i.e., local search). We separately examine how the MAP landscape changes with an increase in exploitation, and next, how the landscape adjusts with two alternative ways of increasing exploration. In fact, the first subsection compares the landscapes of the original VLSN search with the new VNS variant introduced in this paper. VNS allows for greater exploration, and leads to landscape graphs with greater number of local minima. Next, the second subsection compares the landscapes of two multi-start strategies for the VLSN search, i.e., a random and a deterministic grid-based multi-starts.

We have implemented the algorithms using C++ language, the code was executed on AWS r5a.16xlarge instance with 64 cores and 512 GB of RAM. Each instance of the algorithm was run on a single core.

\subsection{Numerically 
comparing the quality of solutions found 
by the original 
VLSN search and its VNS 
}\label{sec:compareVSN2VLSN}

We compare the solution quality of the original 
VLSN search (see Alg.~\ref{alg:VLSN.OP}) and the new variable neighborhood search (VNS) for the VLSN (see Alg.~\ref{alg:VLSN.VNS}).
We expect VNS (Alg.~\ref{alg:VLSN.VNS}) to search through a greater number of solutions and to find higher quality solutions than the original VLSN search (Alg.~\ref{alg:VLSN.OP}). 
Note that the solution quality tends to improve 
as the total number of solutions examined increases. 
%
In fact, the larger amount of solution space is explored, the more likely is that an algorithm finds a better solution. 
Search-based metaheuristics like the VLSN search alternate between the exploitation phase (i.e., search of the neighborhoods) and the exploration phase (i.e., restarting in a new initial solution of the multi-start). 
Because VNS allows to escape a local minimum by expanding search to larger and larger neighborhoods,  Alg.~\ref{alg:VLSN.VNS} stays longer in the exploitation (or search) phase. Whereas Alg.~\ref{alg:VLSN.OP} switches from exploitation (or search) to exploration (or restart) 
sooner %
than Alg.~\ref{alg:VLSN.VNS}. This difference between algorithms and the fact that the number of solutions found in exploitation phase differs with each MAP instance imply that there will be differences in the total number of solutions two algorithms have found  by the time 
they terminate in local minima in the last exploitation phase (i.e., all chosen multi-starts have been explored). 
Therefore, we set the number of restarts $\mu$ to one in both Alg.~\ref{alg:VLSN.OP} and~\ref{alg:VLSN.VNS} ($\mu_{\mathrm{VNS}} = \mu_{\mathrm{VLSN}} = \mu = 1$), and use the same MAP instances and the same starting solutions.

To better understand 
the benefits and trade-offs due to %
greater exploitation of VNS (Alg.~\ref{alg:VLSN.VNS}) and 
compared to %
VLSN (Alg.~\ref{alg:VLSN.OP}), we perform a statistical comparison of the quality of solutions found by these two algorithms. Specifically, we compare the single-start 
VLSN search to the single-start VNS for the VLSN on sets of $I = 100$ randomly generated problem instances for several combinations of cardinality $N$
and dimensionality $D$ 
of the MAP. Specifically, we consider $D=4, N=5,6,7,8,9,10$ and $D=5,N=5$ to lower computational burden of our simulations that involve computing VNS with all order-$K$ neighborhoods.

Let $A_1$ denote the single-start 
VLSN search  (Alg.~\ref{alg:VLSN.OP}) with $\mu=1$, and let $A_2$ be the single-start VNS for the VLSN  (Alg.~\ref{alg:VLSN.VNS}) with $\mu = 1$ and the neighborhood that includes any combinations of $K$ dimensions, where $K=1,2,\ldots,D-K$. That is, we consider the VNS with the neighborhood that is a union of all order-$K$ neighborhoods. 
We denote $y_1, y_2$  the objective values  that are found by algorithms $A_1,A_2$, respectively. 
Similarly, ${\nu}_1,{\nu}_2$ and $\mu_1,\mu_2$ denote the respective numbers of solutions (i.e., nodes in the landscape graph) explored and the corresponding numbers of local minima (i.e., sinks in the landscape graph) found by $A_1,A_2$.
We consider a difference $\Delta y := y_2 - y_1$ between their returned objective values (alternatively, $\Delta\nu:={\nu}_2 - {\nu}_1$ is a difference in the number of searched solutions between $A_1$ and $A_2$; $\Delta\mu=\mu_2 - \mu_1$ are the difference in local minima). The difference~$y_2 - y_1$ is negative, when the quality of solution $y_2$ found by the VNS ($A_2$)  is better (i.e., smaller) than its alternative, multi-start VLSN search $A_1$. For some of $I = 100$ randomly generated problem instances, the~$y_2 - y_1$ differences are randomly distributed. Recall that the objective values $y_1,y_2$ are sums of randomly generated cost coefficients. So, by the law of large numbers, for large cardinality $n$, the difference $y_2 - y_1$ is approximately normally distributed as a difference of approximately normal $y_i, i=1,2$. 

The results of the statistical comparison of the quality of best solutions found by $A_1,A_2$ are summarized in Table~\ref{exp-3AP}. The average difference $y_2-y_1$ is negative, so on average the new VNS outperforms in terms of solution quality. However, for the MAP with small cardinality $N$ (e.g., $N=4,5$ for $D=4$ and when $N=5,D=5$), this out-performance does not seem to be statistically significant. Our results suggest that for any dimensionality $D$, there is some cardinality $N_0 = N_0(D)$ such that the new VNS $A_2$ likely finds the solution with quality $y_2$ that is statistically significantly better than the solution quality $y_1$ of the VLSN search $A_1$.

\begin{table}[]
\centering
\caption{Statistical Comparison of Solution Quality for 
$D=4,5$}
\label{exp-3AP}
\begin{tabular}{|l|l|l|l|l|}
\hline
  \multirow{2}{*}{\textbf{D}} & \multirow{2}{*}{\textbf{N}} & \textbf{Mean} 
 & \multirow{2}{*}{\textbf{Std. $D^y_{2,1}$}} & \multirow{2}{*}{\textbf{$2\sigma $ C.I. $y_{2} - y_{1}$}} \\
  & & $y_{2}-y_{1}$ &   & \\
 \hline
%
%
\hline
4  & 5 & -60192.61 & 64246.29  & (-188685.18, 68299.96) \\
4  & 6 & -81446.61 & 59449.21  & (-200345.03, 37451.81) \\
4  & 7 & -102845.4 & 45350.48  & (-193546.35, -12144.45) \\
4  & 8 & -112645.51 & 44745.47  & (-202136.45, -23154.57) \\
4  & 9 & -128547.98 & 36660.48  & (-201868.94, -55227.02) \\
4  & 10  & -136718.16& 35135.01 & (-206988.18, -66448.14)\\
\hline
5  & 5 & -59066.17  & 36328.95  & (-131724.06, 13591.72)\\

\hline
\hline
%
\end{tabular}
\end{table}

As mentioned above, we expect VNS to outperform VLSN search because VNS searches through more solutions $\nu_2$ than $\nu_1$ of VLSN search. Hence, we examine the trade-off $\Delta y / \Delta\nu$ between an improvement in the solution quality ($\Delta y$) and an increase in the number of solutions ($\Delta\nu$), as defined below
$$\textrm{tradeoff} =\frac{\Delta y}{\Delta \nu} = \frac{y_2 - y_1}{ \nu_2 -\nu_1},$$
where $\nu$ - number of solutions. Similarly, we also consider a trade-off $\Delta y/\Delta\mu$ between improvement in the quality ($\Delta y$) and change in the number of sinks ($\Delta\mu$). 

The results are displayed in Figures~\ref{exp-3AP-tradeoff-sols} and~\ref{exp-3AP-tradeoff-mins}. Because the corresponding differences $\Delta\nu$ and $\Delta\mu$ in solutions and minima are so large between $A_1$ and $A_2$, the mean trade-offs in solution quality for both total solutions searched and local minima quickly becomes smaller in absolute value as $N$ grows.
\begin{table}[]
\centering
\caption{Statistical Comparison of Trade-off $\frac{\Delta y}{\Delta\mu}$ between Solution Quality $y$ and Number of Found Solutions $\nu$ for 
$D=4,5$}
\label{exp-3AP-tradeoff-sols}
\begin{tabular}{|l|l|l|l|l|}
\hline
  \multirow{2}{*}{\textbf{D}} & \multirow{2}{*}{\textbf{N}} & \textbf{Mean} 
 & \multirow{2}{*}{\textbf{Std. $D^{\Delta y/\Delta\nu}_{1,2}$}} & \multirow{2}{*}{\textbf{$2\sigma $ C.I. $\frac{y_{1}-y_{2}}{{\nu}_{1}-{\nu}_{2}}$}} \\
  & & $\frac{y_{1}-y_{2}}{{\nu}_{1}-{\nu}_{2}}$ &   &  \\
 \hline
%
\hline
4 & 5 & -143.55 & 576.93 & (-1297.41, 1010.31) \\
4 & 6 & -18.10 & 41.25 & (-100.60, 64.39) \\
4 & 7 & -3.08 & 3.22 & (-9.52, 3.36) \\
4 & 8 & -0.74 & 0.94 & (-2.61, 1.13) \\
4 & 9 & -0.30 & 0.67 & (-1.63, 1.04) \\
4 & 10 & -0.10 & 0.12 & (-0.35, 0.15) \\
\hline
5 & 5 & -0.36 & 2.10 & (-4.56, 3.84) \\

\hline
\hline
%
\end{tabular}
\end{table}

\begin{table}[]
\centering
\caption{Statistical Comparison of Trade-off $\frac{\Delta y}{\Delta\mu}$ between Solution Quality $y$ and Number of Found Minima (Sinks) $\mu$ for 
$D=4,5$}
\label{exp-3AP-tradeoff-mins}
\begin{tabular}{|l|l|l|l|l|}
\hline
  \multirow{2}{*}{\textbf{D}} & \multirow{2}{*}{\textbf{N}} & \textbf{Mean} 
 & \multirow{2}{*}{\textbf{Std. $D^{\Delta y/\Delta\mu}_{1,2}$}} & \multirow{2}{*}{\textbf{$2\sigma $ C.I. $\frac{y_{1}-y_{2}}{{\mu}_{1}-{\mu}_{2}}$}} \\
  & & $\frac{y_{1}-y_{2}}{{\mu}_{1}-{\mu}_{2}}$ &   &  \\
 \hline
%
%
\hline
4 & 5 & -853.22 & 2164.10 & (-5181.42, 3474.97) \\
4 & 6 & -97.83 & 155.81 & (-409.45, 213.78) \\
4 & 7 & -13.79 & 12.08 & (-37.96, 10.38) \\
4 & 8 & -2.73 & 3.11 & (-8.96, 3.50) \\
4 & 9 & -0.98 & 2.13 & (-5.24, 3.28) \\
4 & 10 & -0.32 & 0.40 & (-1.13, 0.48) \\
\hline
5 & 5 & -18.15 & 14.12 & (-46.39, 10.09) \\

\hline
\hline
%
\end{tabular}
\end{table}

We finish comparison between the new VNS will all combinations and the original VLSN search by examining three basic features of their landscape graphs, namely the numbers of nodes (i.e., solutions), edges (i.e., neighborhood moves), and sinks (i.e., local minima). Again, we run single-start versions of both algorithms on the same randomly generated set of MAP instances, with $I=100$, and use the same starting solutions for both algorithms. We again study the MAPs with $D=4, N = 5,\ldots,10$ and $D=5,N=5$. For each MAP configuration and each algorithm (i.e., VNS, VLSN), we compute the average, minimum, and maximum numbers of nodes, edges, and sinks.

The results are displayed in Tables~\ref{tab:my-nodes_VNS_VLSN},~\ref{tab:my-edges_VNS_VLSN}, and~\ref{tab:my-sinks_VNS_VLSN}. It is easy to see that VNS produces much large landscape graphs, and the size of these graphs grows much faster for VNS than for VLSN search. For example, for $D=4,N=5$, the VNS has 15 times larger average number of nodes, 10 times larger minimum number of nodes, and 7.96 times larger maximum number of nodes, compared with VLSN search. When the MAP dimensionality grows so that $D=5,N=5$, VNS has 916.75 times larger average number of nodes, 461.94 times larger minimum number of nodes, and 457.88 times larger number of maximum nodes than VLSN search has. When the MAP cardinality increases so that $D=4,N=6$, VNS has 60.2 larger average number of nodes than VLSN search has. And when $N$ increases much 
more so that $D=4,N=10$, VNS has 3,481.97 times larger average number of nodes than VLSN search has. Tables~\ref{tab:my-edges_VNS_VLSN} and~\ref{tab:my-sinks_VNS_VLSN} shows similar patterns in relations between VNS and VLSN search for edges, and for sinks, respectively.

\begin{table}[]
\caption{The minimum, maximum, and average number of nodes in the landscape graphs generated by the VNS with all combinations of permuted dimensions and the original VLSN, where both algorithms are initialized from the same one initial starting solution. }
\label{tab:my-nodes_VNS_VLSN}
\begin{tabular}{|l|l|lll|lll|}
\hline
   \multirow{2}{*}{$D$} & \multirow{2}{*}{$N$}
   & \multicolumn{3}{l|}{VNS with all combinations}
  & \multicolumn{3}{l|}{VLSN}                                                                                                                                                      \\ \cline{3-8}
  & & min nodes & max nodes  & avg nodes  & min nodes & max nodes & avg nodes \\ \hline
4 & 5    & 73 & 3138  & 1513.61 & 7       & 394       & 99.73      \\ 
4 & 6    & 405
& 20089     
& 8996.46
& 22        
& 370       
& 149.44   
\\ 
4 & 7    & 4850      
& 152765    
& 51305.72   
& 26        
& 685       
& 217.72    
\\ 
4 & 8    & 24170     
& 878112    
& 251745.79  
& 54        
& 857       
& 362.24    
\\ 
4 & 9    & 23011     
& 4060123   
& 920079.65  
& 96        
& 1671      
& 527.46    
\\ 
4 & 10   & 98029    
& 7423584  
& 2427666.63
& 104       
& 1806      
& 697.21  
\\ \hline
5 & 5    & 7391      
& 697530    
& 471210.28  
& 16        
& 1523      
& 514.47    
\\ \hline \hline
\end{tabular}
\end{table}


\begin{table}[] 
\caption{The minimum, maximum, and average number of edges in the landscape graphs generated by the VNS with all combinations of permuted dimensions and the original VLSN, where both algorithms are initialized from the same one initial starting solution. }
\label{tab:my-edges_VNS_VLSN}
\begin{tabular}{|l|l|lll|lll|}
\hline
   \multirow{2}{*}{$D$} & \multirow{2}{*}{$N$}
   & \multicolumn{3}{l|}{VNS with all combinations}
  & \multicolumn{3}{l|}{VLSN}                                                                                                                                                      \\ \cline{3-8}
  & 
& min edges 
& max edges 
& avg edges   
& min edges 
& max edges 
& avg edges 
\\ \hline
4 & 5   
& \multicolumn{1}{l|}{511}       
& \multicolumn{1}{l|}{21966}     
& 10595.27    
& \multicolumn{1}{l|}{28}        
& \multicolumn{1}{l|}{1576}      
& 398.92    
\\ 
4 & 6    
& \multicolumn{1}{l|}{2835}      
& \multicolumn{1}{l|}{140623}    
& 62975.22    
& \multicolumn{1}{l|}{88}        
& \multicolumn{1}{l|}{1480}      
& 597.76    
\\ 
4 & 7    
& \multicolumn{1}{l|}{33950}     
& \multicolumn{1}{l|}{1069355}   
& 359140.04   
& \multicolumn{1}{l|}{104}       
& \multicolumn{1}{l|}{2740}      
& 870.88    
\\ 
4 & 8    
& \multicolumn{1}{l|}{169190}    
& \multicolumn{1}{l|}{6146784}   
& 1762220.53  
& \multicolumn{1}{l|}{216}       
& \multicolumn{1}{l|}{3428}      
& 1448.96   
\\ 
4& 9    
& \multicolumn{1}{l|}{161077}    
& \multicolumn{1}{l|}{28420861}  
& 6440557.55  
& \multicolumn{1}{l|}{384}       
& \multicolumn{1}{l|}{6684}      
& 2109.84   
\\ 
4 & 10   
& \multicolumn{1}{l|}{686203}    
& \multicolumn{1}{l|}{51965088}  
& 16993666.41
& \multicolumn{1}{l|}{416}       
& \multicolumn{1}{l|}{7224}      
& 2788.84    
\\ \hline
5 & 5    
& \multicolumn{1}{l|}{110865}    
& \multicolumn{1}{l|}{10462950}  
& 7068154.2   
& \multicolumn{1}{l|}{80}        
& \multicolumn{1}{l|}{7615}      
& 2572.35   
\\ \hline \hline
\end{tabular}
\end{table}


\begin{table}[]
\caption{The minimum, maximum, and average number of sinks (i.e., local minima) in the landscape graphs generated by the VNS with all combinations of permuted dimensions and the original VLSN, where both algorithms are initialized from the same one initial starting solution. }
\label{tab:my-sinks_VNS_VLSN}
\begin{tabular}{|l|l|lll|lll|}
\hline
   \multirow{2}{*}{$D$} & \multirow{2}{*}{$N$}
   & \multicolumn{3}{l|}{VNS with all combinations}
  & \multicolumn{3}{l|}{VLSN}                                                                                                                                                      \\ \cline{3-8}
 &  & \multicolumn{1}{l}{min sinks} & \multicolumn{1}{l}{max sinks} & avg sinks & \multicolumn{1}{l}{min sinks} & \multicolumn{1}{l}{max sinks} & avg sinks \\ \hline
4 & 5    & \multicolumn{1}{l}{21}        & \multicolumn{1}{l}{194}       & 137.32    & \multicolumn{1}{l}{4}         & \multicolumn{1}{l}{100}       & 32.29     \\ 
4 & 6    & \multicolumn{1}{l}{114}       & \multicolumn{1}{l}{1913}      & 1200.0    & \multicolumn{1}{l}{6}         & \multicolumn{1}{l}{122}       & 50.0      \\ 
4 & 7    & \multicolumn{1}{l}{1452}      & \multicolumn{1}{l}{21281}     & 9851.09   & \multicolumn{1}{l}{12}        & \multicolumn{1}{l}{225}       & 71.72     \\ 
4 & 8    & \multicolumn{1}{l}{7485}      & \multicolumn{1}{l}{167455}    & 61852.07  & \multicolumn{1}{l}{20}        & \multicolumn{1}{l}{269}       & 117.05    \\ 
4 & 9    & \multicolumn{1}{l}{7233}      & \multicolumn{1}{l}{1053356}   & 262394.96 & \multicolumn{1}{l}{29}        & \multicolumn{1}{l}{504}       & 166.82    \\ 
4 & 10   & \multicolumn{1}{l}{30232}     & \multicolumn{1}{l}{2223671}   & 733653.96  & \multicolumn{1}{l}{35}        & \multicolumn{1}{l}{567}       & 216.28    \\ \hline
5 & 5    & \multicolumn{1}{l}{1366}      & \multicolumn{1}{l}{3918}      & 3530.96   & \multicolumn{1}{l}{7}         & \multicolumn{1}{l}{514}       & 181.82    \\ \hline \hline
\end{tabular}
\end{table}

\subsection{Comparison of the landscape graphs for two alternative multi-start strategies for the original VLSN search }
\label{sec:JOGO2020}
We compare the differences in the MAP landscapes generated by two alternative exploration (i.e., multi-start) schemes for the VLSN search. Following the multi-start strategies descriptions in~\cite{kammerdiner2021multidimensional}, we choose to compare the landscapes for a random multi-start (\textbf{VLSNMS}) and a deterministic, grid-based multi-start (\textbf{Grid-VLSN}). Specifically, we run both  multi-start VLSN search implementations with the same number of multi-starts $\mu$, the same initial solutions, on  the same $I=100$ randomly generated instances for each MAP configuration with fixed $N,D$. We study the MAPs with $D=3, N = 10,20,30,40,50,100$ and $D=4,N=4,10,20,30$. The numbers of starting solutions differ based on specific $D,N$, but it is the same $\mu=N^{D-1}$ for both random and grid-based multi-start schemes. For example, the number of starting solutions $\mu=100$ for $D=3,N=10$. 

We show the estimated means (\textbf{avg}) and standard deviations (\textbf{std}) for the numbers of sources (i.e., starting solutions), nodes (i.e., all searched solutions), and sinks (i.e., found local minima). The respective results for sources, nodes, and sinks are presented in Tables~\ref{exp-3_sources},~\ref{exp-3_nodes}, and~\ref{exp-3_sinks}.

As a multi-start search algorithm goes through parts of a set of feasible solutions of a problem instance, the search produces an out-forest graph \cite{KAMMERDINER201478}, where its directed trees has roots or \emph{sources} at constructed starting solutions of the search, leaves or \emph{sinks} at local minimum solutions, and intermediate nodes at the other feasible solutions.  Together Tables~\ref{exp-3_sources},~\ref{exp-3_nodes}, and~\ref{exp-3_sinks} provide information about the graph structure (e.g., sources, nodes, and sinks) of subset of feasible solutions for two types of VLSN search, one using random multi-starts and another utilizing a grid-based multi-starts. As seen in Table~\ref{exp-3_sinks}, the search graphs grow fast with problem parameters and produce many local optima. 
Furthermore, the comparison of these two alternative ways of exploration of the solution set shows that for larger problems, grid-based multi-start tends to find better quality solutions (based on lower average gaps) and to comb through a smaller number of nodes while the number of initiated re-starts is kept the same for both. This suggests that grid-based multi-starts is better able to explore the various parts of the solution set.

\begin{table}[]
\centering
\caption{
Averages and standard deviations for the number of sources (starting solutions) in the VLSN landscape graphs for two alternative multi-start procedures: 
random multi-starts (\textbf{VLSNMS}), and 
grid-based multi-starts (\textbf{Grid-VLSN}). 
We use 100 randomly generated MAP instances  with dimensionality $D$, and cardinality $N$.}
\label{exp-3_sources} 
\begin{tabular}{r|r|rr|rr}
\hline
 \multirow{2}{*}{\textbf{$D$}}   & \multirow{2}{*}{\textbf{$N$}} &  \multicolumn{2}{c|}{\textbf{VLSNMS}} & \multicolumn{2}{c}{\textbf{Grid-VLSN}} \\
    &  &  \textbf{avg sources} & \textbf{std sources} 
    & \textbf{avg sources} & \textbf{std sources} 
    \\ \hline
 3 & 3   & 5.72 & 1.19 
 & 6.04 & 0.95 
 \\
 3 & 10   & 100 & 0.00 
 & 100 & 0.00 
 \\
 3 & 20   & 400 & 0.00 
 & 400 & 0.00 
 \\
 3 & 30   & 900 & 0.00 
 & 900 & 0.00  
 \\
 3 & 40   & 1600 & 0.00 
 & 1600 & 0.00 
 \\
 3 & 50   & 2500 & 0.00 
 & 2500 & 0.00 
 \\
 3 & 100   & 10000 & 0.00 
 & 10000 & 0.00 
 \\ \hline
 4 & 4   & 61.47 & 1.45 
 & 57.08 & 2.40 
 \\
 4 & 10   & 1000 & 0.00 
 & 1000 & 0.00 
 \\
 4 & 20   & 8000 & 0.00 
 & 8000 & 0.00 
 \\
 4 & 30   & 27000 & 0.00 
 & 27000 & 0.00 
 \\
 \hline
 \hline
\end{tabular}
\end{table}


\begin{table}[]
\centering
\caption{
Averages and standard deviations for the number of nodes in the VLSN landscape graphs for two alternative multi-start procedures: 
random multi-starts (\textbf{VLSNMS}), and 
grid-based multi-starts (\textbf{Grid-VLSN}). 
We use 100 randomly generated MAP instances  with dimensionality $D$, and cardinality $N$.}
\label{exp-3_nodes}
\begin{tabular}{r|r|rr|rr}
\hline
 \multirow{2}{*}{\textbf{$D$}}   & \multirow{2}{*}{\textbf{$N$}} &  \multicolumn{2}{c|}{\textbf{VLSNMS}} & \multicolumn{2}{c}{\textbf{Grid-VLSN}} \\
    &  &  \textbf{avg nodes} & \textbf{std nodes} 
    & \textbf{avg nodes} & \textbf{std nodes} 
    \\ \hline
 3 & 3   
 & 15.55 & 1.85 
 &	14.58 & 1.89 
 \\
 3 & 10   
 & 4550.26 & 309.78 
 &	574.13 & 54.45 
 \\
 3 & 20   
 & 66574.26 & 4843.44 
 &	3661.87 & 338.61 
 \\
 3 & 30   
 & 355903.73 & 17745.00 
 &	12742.86 & 1009.61 
 \\
 3 & 40   
 & 1319062.92 & 63481.93 
 &	34630.15 & 2465.55 
 \\
 3 & 50  
 & 269942.47 & 236.50 
 &	81231.99 & 5624.46 
 \\
 3 & 100   
 & 530012.67 & 729.49 
 &	401795.20 & 66.62 
 \\ \hline
 4 & 4   
 & 650.85 & 28.65 
 &	337.42 & 27.00 
 \\
 4 & 10   
 & 213312.35 & 57049.72 
 &	69957.31 & 4623.35 
 \\
 4 & 20   
 & 1701591.17 & 550.25 
 &	1604876.12 & 286.66 
 \\
 4 & 30   
 & 3145957.31 & 1369.54 
 &	2733592.89 & 266.00 
 \\
\hline \hline
\end{tabular}
\end{table}


\begin{table}[]
\centering
\caption{
Averages and standard deviations for the number of sinks (local minima) in the VLSN landscape graphs for two alternative multi-start procedures: 
random multi-starts (\textbf{VLSNMS}), and 
grid-based multi-starts (\textbf{Grid-VLSN}). 
We use 100 randomly generated MAP instances  with dimensionality $D$, and cardinality $N$.}
\label{exp-3_sinks}
\begin{tabular}{r|r|rr|rr}
\hline
 \multirow{2}{*}{\textbf{$D$}}   & \multirow{2}{*}{\textbf{$N$}} &  \multicolumn{2}{c|}{\textbf{VLSNMS}} & \multicolumn{2}{c}{\textbf{Grid-VLSN}} \\
    &  &  \textbf{avg sinks} & \textbf{std sinks} 
    & \textbf{avg sinks} & \textbf{std sinks} 
    \\ \hline
 3 & 3   
 &   1.58 & 0.64 
 & 1.58 & 0.64
 \\
 3 & 10   
 &   1167.98 & 59.29 
 &139.86 & 13.10
 \\
 3 & 20   
 &   17539.51 & 1141.31 
 &864.34 & 86.46
 \\
 3 & 30   
 &   88339.38 & 4132.30 
 &2944.43 & 246.26
 \\
 3 & 40   
 &   315121.06 & 14303.08 
 &7901.09 & 567.36
 \\
 3 & 50  
 &   54276.75 & 249.87 
 &18326.57 & 1261.10
 \\
 3 & 100   
 &   58713.34 & 267.18 
 & 84038.32 & 231.55
 \\ \hline
 4 & 4   
 &   49.85 & 9.60 
 &43.71 & 7.60
 \\
 4 & 10   
 &   63237.82 & 17808.02 
 &21524.84 & 1384.00
 \\
 4 & 20   
 & 428071.05 & 1115.62 
 
 & 436636.17 & 1146.15
 \\
 4 & 30   
 &   637265.08 & 1022.70 
 & 682608.79 & 1164.37
 \\
\hline \hline
\end{tabular}
\end{table}

\section{Path length analysis and the fitness-distance correlations}\label{sec:correl}
In this section, we numerically analyze the fitness-distance correlations. Because of the relation with problem instance difficulty~\cite{jones1995fitness}, a major theme in the work on combinatorial landscapes is analysis of the fitness-distance correlation. Specifically, we study the relation between the fitness and distance for the local minima of the MAP instance (or the sinks of the MAP landscape of a given problem instance). For any feasible solution, including a local minimum, the fitness refers to the objective value of this solution. While the distance is defined here as the length of a shortest path to this solution from some given starting solution. 

We perform computational experiments as follows. As in Section~\ref{sec:compareVSN2VLSN}, we consider a sequence of the seven MAPs with parameter pairs 
$(D,N)$ (sometimes denoted $D\times N$), where either $D=4$ and $N = 5,6,7,8,9$, or $D=5, N=5$. 
For each MAP with fixed $D$ and $N$, we generate $I=100$ instances with randomly chosen cost coefficients as described in detail in~\cite{kammerdiner2021multidimensional}.
We solve each MAP instance using either the VNS or the VLSN search. Again, just as in Section~\ref{sec:compareVSN2VLSN}, (i) both algorithms are only started once, (ii) this starting solution is the same for both algorithms, and (iii) the same $I$ problem instances are used for running both algorithms.  
Before calculating and examining the fitness-distance correlations, we first compute and analyze path lengths.

The path lengths analysis involves first finding a/the shortest path from the (single) given starting solution to each local minimum (or a sink in the landscape graph) for every instance, and then calculating the lengths of these shortest paths. Hence, for each MAP $\mathcal{P}_j$ with fixed $D_j,N_j, j=1,\ldots,7$ and each instance $P_i, i=1,\ldots, I$ of $\mathcal{P}$, we have a total of $m_{ij}$ local minima (sinks), and each minimum having the respective path lengths $l(m_{ij})$. 

Figures~\ref{fig:MAP_pathsVNS} and~\ref{fig:MAP_pathsVLSN} display box plots of the path length distribution for the VNS and the VLSN search, respectively. Each figure shows seven box plots for seven different MAP configurations $D\times N$. When comparing the box plots on the two figures, several observations are apparent. First, the new VNS has a significantly steeper increase \emph{with respect to cardinality} $N$ in the path lengths to a local minimum than the original VLSN search. Second, the path length of the VLSN search does not appear significantly different. Indeed, while the average lengths of the VLSN search do seem to grow, but rather moderately.  Third, the variation in path lengths also increases \emph{with respect to cardinality} $N$ for the VNS, but not necessary for the VLSN search. This suggests that for MAP instances with greater $N$, the VNS is able to keep searching much longer than the original VLSN search, resulting in deeper exploitation.

\begin{figure}
  \centering
  \includegraphics[angle=0,width=0.9\textwidth]{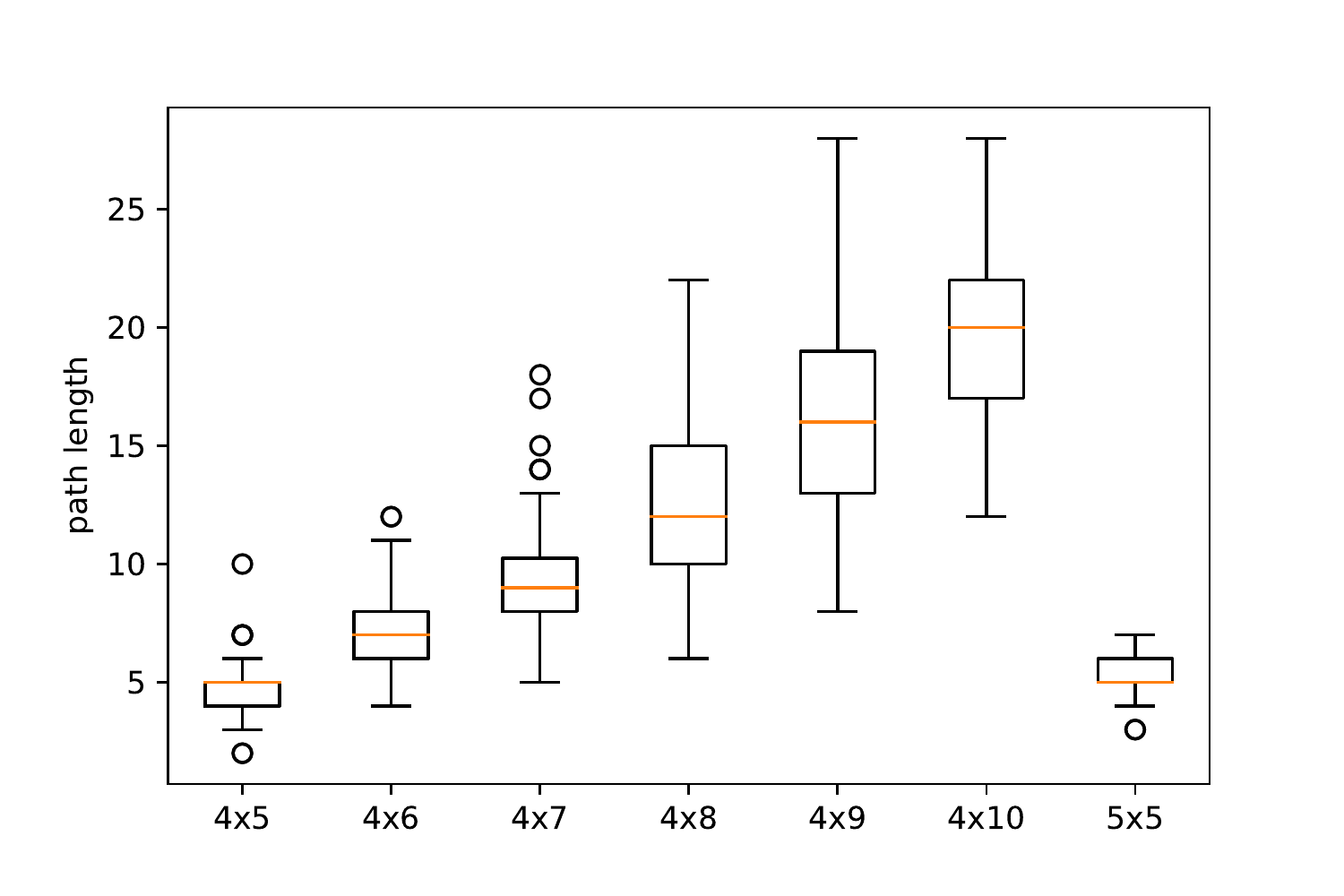}\\
  \caption{Box plot showing length of path to the optimal solution, VNS}
  \label{fig:MAP_pathsVNS}
\end{figure}

\begin{figure}
  \centering
  \includegraphics[angle=0,width=0.9\textwidth]{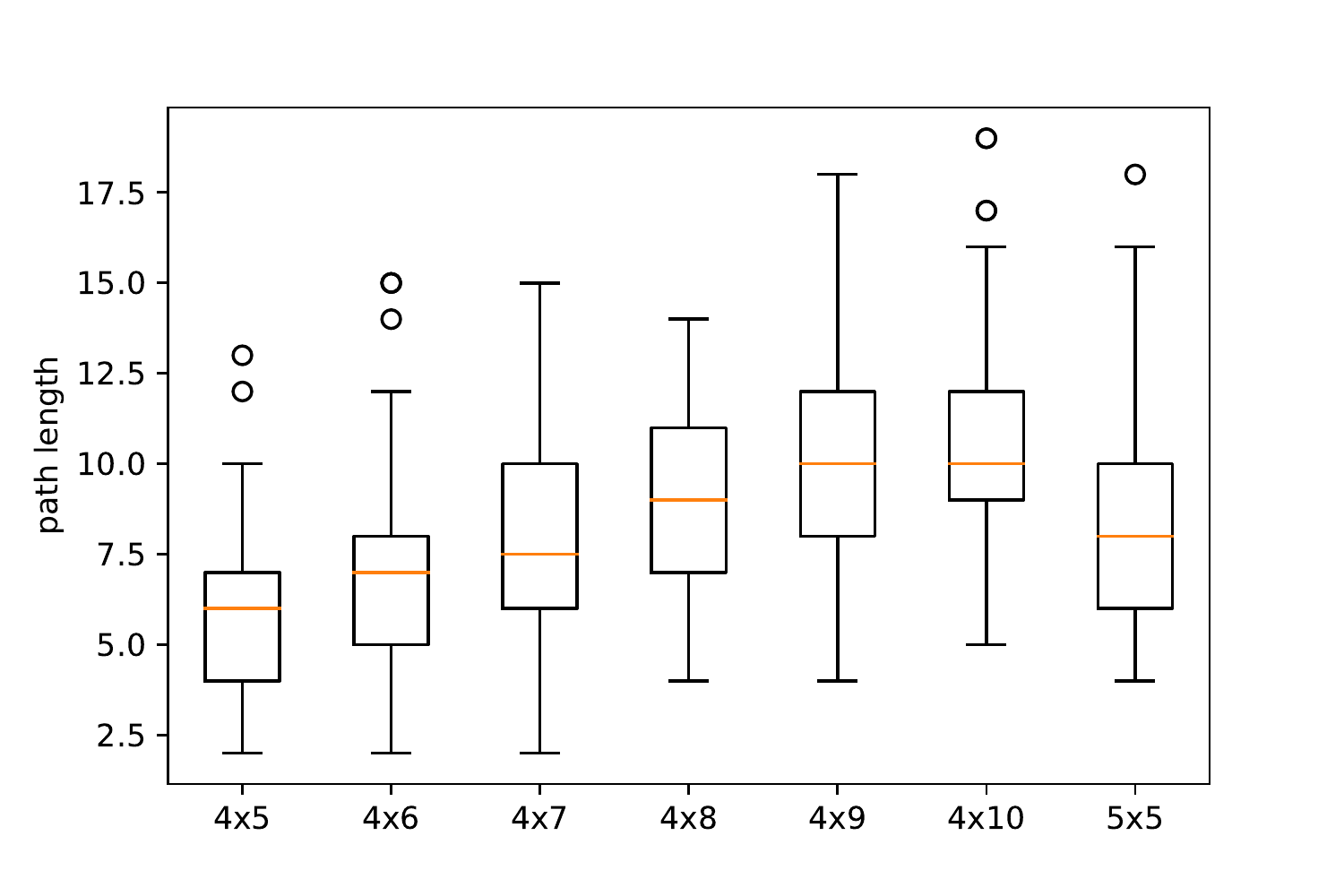}\\
  \caption{Box plot showing length of path to the optimal solution, VLSN}
  \label{fig:MAP_pathsVLSN}
\end{figure}

Next, we calculate the fitness-distance correlations among the local minima similarly to the way we computed the path length. Specifically, we fix the MAP with given parameter pair $D$ and $N$. For each MAP, we consider the same $I=100$ randomly generated instances. Each instance was solved using the VLSN search and the VNS to find all local minima (sinks of the landscape). For each local minimum $m$, we use the previously computed path lengths $l(m)$ as the \emph{distance}. For each local minimum $m$, we also find its objective value $f(m)$, which gives the \emph{fitness}. Finally, for each instance, we compute the correlation $\rho = corr(l(m),f(m))$ between the distance and fitness values of all local minima of this problem instance.

The results are shown in Figures~\ref{fig:MAP_corrVNS} and~\ref{fig:MAP_corrVLSN}, which display box plots of the fitness-distance correlations for the VNS and the VLSN search, respectively. Each respective figure shows seven or six box plots for 
different MAP configurations $D\times N$. When comparing the box plots on the two figures, several observations are apparent. First, the new VNS has a significant steep decrease \emph{with respect to cardinality} $N$ in the fitness-distance correlations to a local minimum.  Second, the correlation of the original VLSN search does not change significantly with respect to cardinality $N$. %
Third, the variation in correlations significantly decreases \emph{with respect to cardinality} $N$ for the VSN, but this decrease is not significant at all for the VLSN search. 
Following reasoning analogous to~\cite{jones1995fitness}, if we are minimizing, we would like to see the fitness decrease as the distance from the starting solution to local minimum increases. Hence, the ideal fitness function will have $\rho = -1.0$ 
This suggests that for MAP instances with smaller $N$, the VNS produces a ruggeder landscape than the original VLSN search. In other words, there may be a trade-off in using the VNS for the MAPs with smaller cardinality $N$.

\begin{figure}
  \centering
  \includegraphics[angle=0,width=0.9\textwidth]{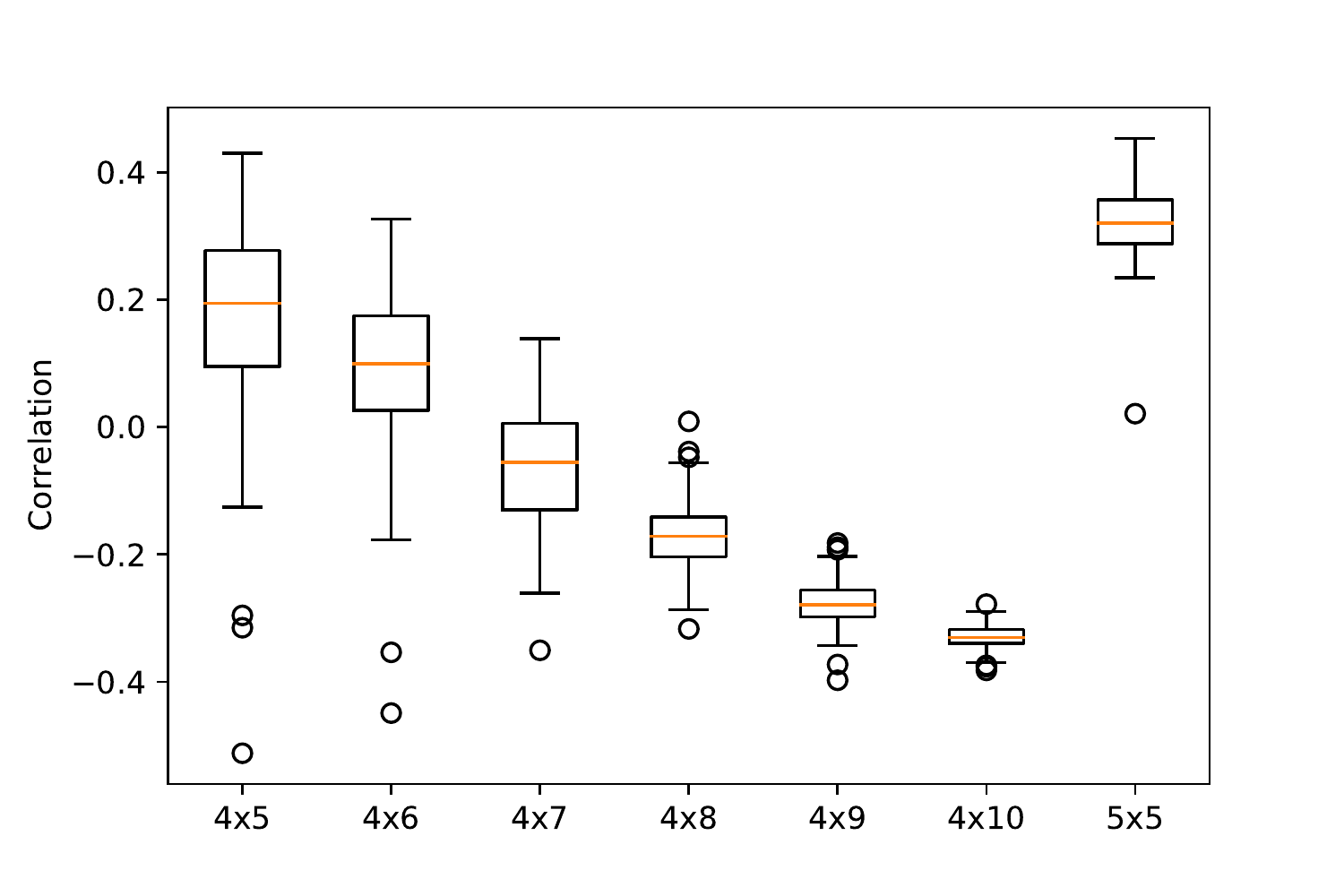}\\
  \caption{Correlations between fitness and distance for the local minima, VNS}
  \label{fig:MAP_corrVNS}
\end{figure}

\begin{figure}
  \centering
  \includegraphics[angle=0,width=0.9\textwidth]{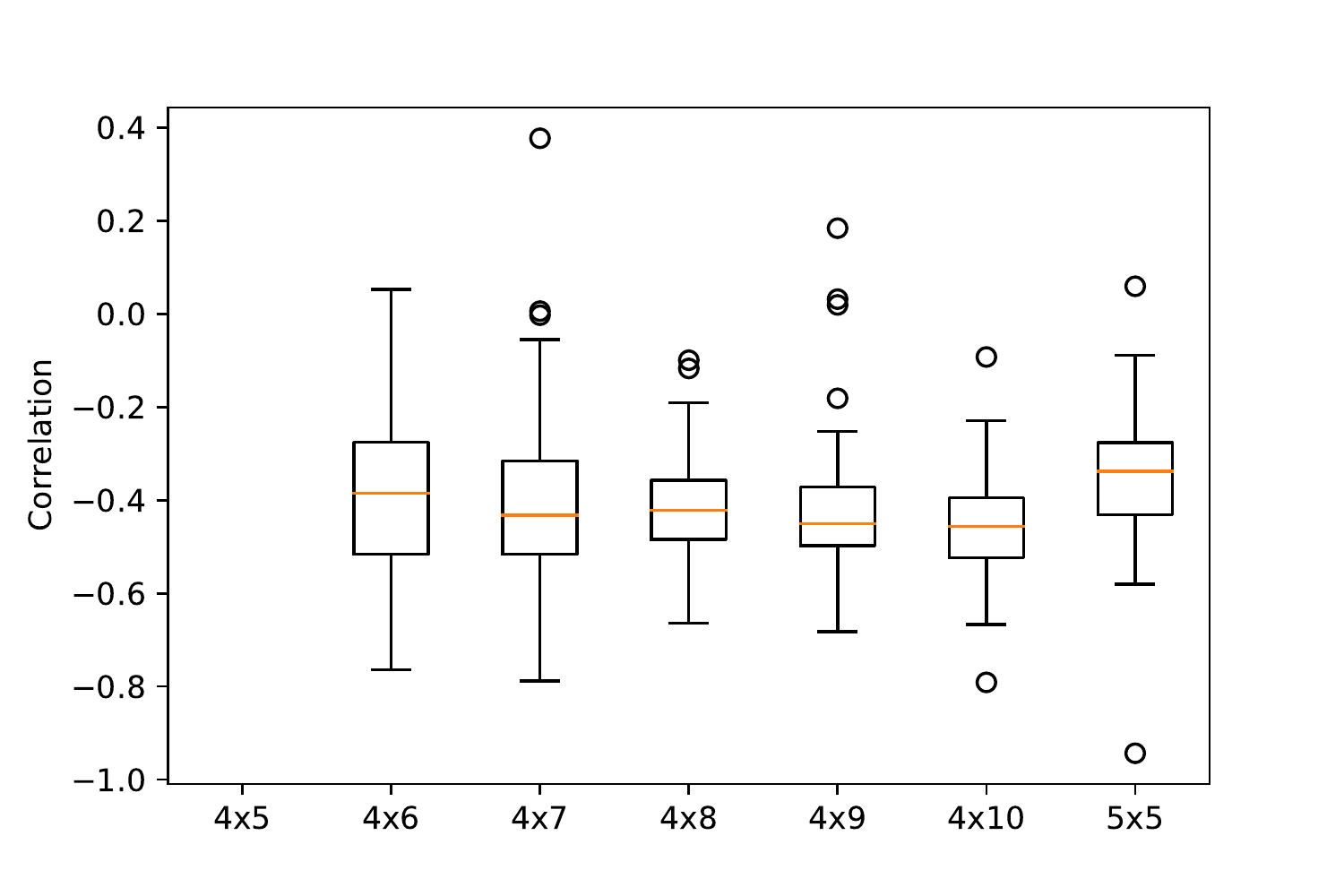}\\
  \caption{Correlations between fitness and distance for the local minima, VLSN}
  \label{fig:MAP_corrVLSN}
\end{figure}

\section{Conclusion}
\label{sec:end}
We conduct numerical and theoretical analyses of the MAP landscapes generated by various implementations of VLSN search. Furthermore, we propose the VNS algorithm, a novel version of the VLSN search that utilizes variable neighborhoods of order $K$, $K<D$. This new version enables further exploitation and results in a larger number of local minima than the original VLSN search. In terms of its solution quality, VNS statistically significantly outperforms the original VLSN for all MAP configurations except the MAPs with very small cardinality $N$. 

Moreover, we show by induction on $D$ that an undirected graph of the MAP landscape for $N=2$ can be represented by a hypercube. Further research should be aimed at generalizing the result for $N>2$ and determining the digraph structure of the MAP landscapes produced by VLSN search.

We conduct computational experiments to compare the solution quality, and the landscape graphs' structures between alternative versions of the VLSN search for the MAP. We are able to separately examine differences in exploitation versus differences in exploration. In particular, we study the benefits and trade-offs of greater exploitation by comparing and contrasting VNS with VLSN search. To understand impact of the exploration phase on the MAP landscapes, we consider original VLSN search and compare two different multi-start strategies, namely a random and a deterministic, grid-based.
Additionally, we analyze the path lengths and the fitness-distance correlations for VNS and VLSN search.
The findings in this paper can be used to design improved VLSN-based solution algorithms for the MAP.

\section*{Acknowledgment}
A. Kammerdiner was supported by the AFRL (National Research Council Fellowship). 

\section*{Data availability statement}

Because of the size of the simulated data, it is neither practical nor necessary to share all data that we generated or analysed during this study. 
The minimal dataset that would be necessary to interpret, and build upon the findings reported in the article consists of numerous tables and figures, all of which are included in this published article. For purpose of the replication, the authors will make the code used to simulate the data available via github.

\bibliographystyle{abbrv}
\bibliography{ref_wcgo_jogo_si} 

\begin{thebibliography}{10}

\bibitem{agaev2005spectra}
R.~Agaev and P.~Chebotarev.
\newblock On the spectra of nonsymmetric laplacian matrices.
\newblock {\em Linear Algebra and its Applications}, 399:157--168, 2005.

\bibitem{agaev2000matrix}
R.~P. Agaev and P.~Y. Chebotarev.
\newblock The matrix of maximal outgoing forests of a digraph and its
  applications.
\newblock {\em Avtomatika i Telemekhanika}, (9):15--43, 2000.

\bibitem{altenberg1997fitness}
L.~Altenberg.
\newblock Fitness landscapes: Nk landscapes.
\newblock {\em Handbook of Evolutionary Computation}, pages 5--10, 1997.

\bibitem{bosman2020visualising}
A.~S. Bosman, A.~Engelbrecht, and M.~Helbig.
\newblock Visualising basins of attraction for the cross-entropy and the
  squared error neural network loss functions.
\newblock {\em Neurocomputing}, 400:113--136, 2020.

\bibitem{burkard1999linear}
R.~E. Burkard and E.~Cela.
\newblock Linear assignment problems and extensions.
\newblock In {\em Handbook of combinatorial optimization}, pages 75--149.
  Springer, 1999.

\bibitem{choromanska2015loss}
A.~Choromanska, M.~Henaff, M.~Mathieu, G.~B. Arous, and Y.~LeCun.
\newblock The loss surfaces of multilayer networks.
\newblock In {\em Artificial intelligence and statistics}, pages 192--204.
  PMLR, 2015.

\bibitem{choromanska2015open}
A.~Choromanska, Y.~LeCun, and G.~B. Arous.
\newblock Open problem: The landscape of the loss surfaces of multilayer
  networks.
\newblock In {\em Conference on Learning Theory}, pages 1756--1760. PMLR, 2015.

\bibitem{daolio2011communities}
F.~Daolio, M.~Tomassini, S.~V{\'e}rel, and G.~Ochoa.
\newblock Communities of minima in local optima networks of combinatorial
  spaces.
\newblock {\em Physica A: Statistical Mechanics and its Applications},
  390(9):1684--1694, 2011.

\bibitem{hordijk2005correlation}
W.~Hordijk and S.~A. Kauffman.
\newblock Correlation analysis of coupled fitness landscapes.
\newblock {\em Complexity}, 10(6):41--49, 2005.

\bibitem{jones1995fitness}
T.~Jones, S.~Forrest, et~al.
\newblock Fitness distance correlation as a measure of problem difficulty for
  genetic algorithms.
\newblock In {\em ICGA}, volume~95, pages 184--192, 1995.

\bibitem{kammerdiner2013application}
A.~Kammerdiner and E.~Pasiliao.
\newblock Application of graph-theoretic approaches to the random landscapes of
  the three-dimensional assignment problem.
\newblock {\em Optimization Letters}, 7(1):79--87, 2013.

\bibitem{KAMMERDINER201478}
A.~Kammerdiner and E.~Pasiliao.
\newblock In and out forests on combinatorial landscapes.
\newblock {\em European Journal of Operational Research}, 236(1):78 -- 84,
  2014.

\bibitem{kammerdiner2021multidimensional}
A.~Kammerdiner, A.~Semenov, and E.~Pasiliao.
\newblock Multidimensional assignment problem for multipartite entity
  resolution, 2021.

\bibitem{kammerdiner2017very}
A.~Kammerdiner and C.~F. Vaughan.
\newblock Very large-scale neighborhood search for the multidimensional
  assignment problem.
\newblock In S.~Butenko, P.~M. Pardalos, and V.~Shylo, editors, {\em
  Optimization Methods and Applications}. Springer, 2017.

\bibitem{kammerdiner2009multidimensional}
A.~R. Kammerdiner.
\newblock Multidimensional assignment problem.
\newblock In {\em Encyclopedia of Optimization}, pages 2396--2402. Springer,
  2009.

\bibitem{karp1972reducibility}
R.~M. Karp.
\newblock Reducibility among combinatorial problems.
\newblock In {\em Complexity of computer computations}, pages 85--103.
  Springer, 1972.

\bibitem{malan2021survey}
K.~M. Malan.
\newblock A survey of advances in landscape analysis for optimisation.
\newblock {\em Algorithms}, 14(2):40, 2021.

\bibitem{merz2004advanced}
P.~Merz.
\newblock Advanced fitness landscape analysis and the performance of memetic
  algorithms.
\newblock {\em Evolutionary Computation}, 12(3):303--325, 2004.

\bibitem{ochoa2009analyzing}
G.~Ochoa, R.~Qu, and E.~K. Burke.
\newblock Analyzing the landscape of a graph based hyper-heuristic for
  timetabling problems.
\newblock In {\em Proceedings of the 11th Annual conference on Genetic and
  evolutionary computation}, pages 341--348, 2009.

\bibitem{pierskalla1967tri}
W.~P. Pierskalla.
\newblock The tri-substitution method for the three-dimensional assignment
  problem.
\newblock {\em CORS Journal}, 5:71--81, 1967.

\bibitem{pierskalla1968letter}
W.~P. Pierskalla.
\newblock Letter to the editor—the multidimensional assignment problem.
\newblock {\em Operations Research}, 16(2):422--431, 1968.

\bibitem{reidys2002combinatorial}
C.~M. Reidys and P.~F. Stadler.
\newblock Combinatorial landscapes.
\newblock {\em SIAM review}, 44(1):3--54, 2002.

\bibitem{richter2008coupled}
H.~Richter.
\newblock Coupled map lattices as spatio-temporal fitness functions: Landscape
  measures and evolutionary optimization.
\newblock {\em Physica D: Nonlinear Phenomena}, 237(2):167--186, 2008.

\bibitem{schiavinotto2007review}
T.~Schiavinotto and T.~St{\"u}tzle.
\newblock A review of metrics on permutations for search landscape analysis.
\newblock {\em Computers \& operations research}, 34(10):3143--3153, 2007.

\bibitem{stadler1992landscape}
P.~F. Stadler and W.~Schnabl.
\newblock The landscape of the traveling salesman problem.
\newblock {\em Physics Letters A}, 161(4):337--344, 1992.

\end{thebibliography}
\end{document}